\documentclass[11pt,oneside,english]{article}

\usepackage[T1]{fontenc}
\usepackage[latin9]{inputenc}
\usepackage{varioref}
\usepackage{amsthm}
\usepackage{amssymb}
\usepackage{comment}
\usepackage[usenames,dvipsnames]{xcolor}

\usepackage{cmap}
\usepackage{hyperref}
\usepackage[active]{srcltx}

\usepackage{mathtools}

\usepackage[shortlabels]{enumitem}

\makeatletter
\numberwithin{equation}{section} 
\numberwithin{figure}{section} 
\theoremstyle{plain}
\theoremstyle{plain}
\newtheorem{thm}{Theorem}
  \theoremstyle{plain}
  \newtheorem{lem}{Lemma}
  \theoremstyle{plain}
  \newtheorem{prop}{Proposition}
  \theoremstyle{remark}
  \newtheorem*{note*}{Note}
  \theoremstyle{remark}
  \newtheorem*{conclusion*}{Conclusion}
  \theoremstyle{remark}
  \newtheorem{note}{Remark}
 \theoremstyle{definition}
  
  \newtheorem{defc}{Definition}
  \theoremstyle{plain}
  \newtheorem{cor}{Corollary}

\usepackage{amsthm}
\usepackage{amsfonts}
\usepackage{amscd}
\usepackage[cmtip,arrow]{xy}
\usepackage{pb-diagram,pb-xy}

\usepackage{tikz}
\usetikzlibrary{matrix,arrows}

\usepackage{tikz-cd}
\usetikzlibrary{decorations.pathreplacing, positioning}

\usepackage{soul}

\usepackage{cancel}

\newcommand{\g}{\mathfrak{g}}

\newcommand{\kf}{\mathfrak{k}}

\newcommand{\pf}{\mathfrak{p}}

\newcommand{\dif}{\textrm{\textbf{d}}}

\newcommand{\cI}{{\mathcal I}}

\newcommand{\cL}{{\mathcal L}}

\newcommand{\mR}{\mathbb{R}}
\newcommand{\mC}{\mathbb{C}}

\newcommand{\wedcol}{\,\overset{\scriptscriptstyle\wedge}{,}\,}
\newcommand{\weddot}{\,\overset{\scriptscriptstyle\wedge}{\cdot}\,}

\usepackage[e]{esvect}

\usepackage{charter}

\makeatother

\begin{document}

\title{\sc Griffiths variational multisymplectic formulation for Lovelock gravity}
\author{{S. Capriotti}\thanks{(Corresponding author) Dept. Matem\'atica, Univ. Nacional del Sur, 8000 Bah\'ia Blanca, Argentina. {\it e-mail}: santiago.capriotti@uns.edu.ar.},
{J. Gaset}\thanks{Dept. Physics, Univ. Aut\`onoma de Barcelona, 08193 Bellaterra, Spain. {\it e-mail}: jordi.gaset@uab.cat.},
{N. Rom\'an-Roy}\thanks{Dept. Mathematics.
   Univ. Polit\`ecnica de Catalunya.
   08034 Barcelona, Spain. {\it e-mail}: narciso.roman@upc.edu.}, 
{L. Salomone}\thanks{Dept. Matem\'atica-CMaLP, Facultad de Ciencias Exactas, UNLP, 1900 La Plata, and CONICET, Argentina. {\it e-mail}: salomone@mate.unlp.edu.ar}}
\date{\today}

\maketitle

\begin{abstract}
This work is mainly devoted to constructing a multisymplectic description 
of Lovelock's gravity, which is an extension of General Relativity.
We establish the Griffiths variational problem for the Lovelock Lagrangian,
obtaining the geometric form of the corresponding field equations.
We give the unified Lagrangian--Hamiltonian formulation of this model
and we study the correspondence between the unified formulations for the 
Einstein--Hilbert and the Einstein--Palatini models of gravity.
\end{abstract}


\noindent\textbf{MSC\,2010 codes:}
{\it Primary}: 49S05, 70S05, 83D05; {\it Secondary}:  35Q75, 35Q76, 53D42, 55R10.

\setcounter{tocdepth}{2}
{
\def\baselinestretch{0.97}
\small
\def\addvspace#1{\vskip 1pt}
\parskip 0pt plus 0.1mm
\tableofcontents
}


\section{Introduction}

The development of the geometric description of classical field theories
using multisymplectic \cite{CCI-91,EMR-96,GMS-97,GIMMSY-mm,art:Roman09,book:Saunders89} or polysymplectic and $k$-(co)symplectic manifolds
\cite{LeSaVi2016,Gu-87,Ka-98}
has rekindled the interest in doing a totally covariant description 
of many theories in physics and, in particular, General Relativity
and other derived from it.
Many general aspects as well as specific problems and characteristics of
the theory have been studied in this way
(see, for instance, \cite{first,CreTe-2016,art:GR-2016,GIMMSY-mm,Ka-13,Ka-q2,Krupka,KrupkaStepanova,rosado2,rovelli}).

In particular, the multisymplectic techniques have been applied
to describe the most standard models of General Relativity:
the {\sl Einstein-Hilbert}  \cite{GRR} and the {\sl Einstein-Palatini} 
(or {\sl metric-affine\/})   \cite{doi:10.1142/S0219887818500445} models
(see, for instance, \cite{Capriotti:2012gg,doi:10.1142/S0219887818500445,GR-2019,vey1}).
In some of these applications, a unified formalism which joins 
the Lagrangian and Hamiltonian formalisms  into a single one has been used.
This unified Lagrangian-Hamiltonian formalism,
introduced for the first time in the pioneering work 
of R. Skinner and R. Rusk \cite{art:Skinner_Rusk83},
is especially useful in mechanics and field theories \cite{CV,dLMM,ELM+,PMRR}
when the Lagrangian that describes the system is singular. 
For this reason, such formalism finds immediate application 
in the study of both, the Einstein--Hilbert and the Einstein--Palatini  models of gravity.
In the first case, and following the formulation for second order field theories
developed in \cite{PMRR}, the symmetrized jet-multimomentum bundle is used as framework,
which turns out to be a premultisymplectic bundle and therefore 
admits the use of the premultisymplectic constraint algorithm \cite{dLMSM,dLMSM+} 
for the study of the field equations. 
In the second case, an indirect path for the construction of unified formalism is invoked:
 first, the field theory corresponding to the Einstein--Palatini model is formulated in \cite{Capriotti:2012gg} as a Griffiths variational problem \cite{book:852048}. Subsequently, and inspired by the work of Gotay \cite{GotayCartan}, a unified formalism is constructed as a Lepage-equivalent problem related to the latter \cite{doi:10.1142/S0219887818500445}. Although it is known that the Einstein--Palatini
and the Hilbert-Einstein Lgrangians essentially lead to the same 
field equations \cite{DP1} (the Einstein equations), the way in which the unified formulations correspond to each other is unknown.

In the last decades, new models that extend General Relativity 
have emerged in theoretical physics \cite{CdL-2011,CGM-16,Fe-12}. In particular,
{\sl Lovelock's gravity} is a generalization of General Relativity (in vacuum) 
introduced by D. Lovelock \cite{Lovelock1970,lovel}
(see also \cite{DF-2012} for a previous work on the canonical analysis of Lanczos-Lovelock gravity).
His idea was to characterize all the symmetric tensors of order 2,
without divergence, that can be constructed from the metric tensor and its derivatives 
up to second order. In dimension 4, it turns out that the only tensors that verify 
these properties are the metric and the Einstein tensor. 
In addition Lovelock proved that this tensor encodes the
Euler--Lagrange equations of a Lagrangian density that is a polynomial 
in the (pseudo) Riemannian curvature.
An interesting characterization for the Lovelock Lagrangian is provided in \cite{DP2}: 
it is the only Lagrangian that is a polynomial in the (pseudo) Riemannian curvature 
and is also stable under a procedure called {\sl consistent Levi-Civita truncation}.
Similar considerations can be found in \cite{P,PK}, where the idea consists in considering
the Lagrangian as a function independent of the metric and the curvature, 
and to find relations between the partial derivatives of the Lagrangian 
with respect to these variables, induced by the geometry of the problem. As the other aforementioned models in General Relativity, 
the Lovelock Lagrangian is singular and then the multisymplectic formulation and, in particular,
the unified Lagrangian-Hamiltonian formalism are especially suitable for its study.

The objectives of this paper are to state and prove the most general and precise results on the following aspects:
to study the correspondence between the unified formulations for the Einstein--Hilbert and 
the Einstein--Palatini models of gravity,
to define the Lovelock Lagrangian in the context of multisymplectic geometry, 
to characterize geometrically its properties, 
to establish the Griffiths variational problem for this Lagrangian
and to develop the corresponding Lagrangian--Hamiltonian unified formalism.

The organization of the paper is as follows:
First, in Section \ref{framebundle}, we set the basic definitions, 
notation and canonical structures of the frame bundle,
which is widely used in the work.
In Section \ref{varLov}, the Lovelock Lagrangian is presented and the
corresponding variational problem is stated and analyzed.
Section \ref{fieqs} is devoted to introduce the infinitesimal symmetries
of the system described by the Lovelock Lagrangian and to obtain
the field equations that derive from the Griffiths variational problem for this system.
Finally, in Section \ref{unified}, the premultisymplectic description of the
Lovelock system is carried out using the unified Lagrangian--Hamiltonian formalism.
After the conclusions of Section \ref{concl}, an Appendix is included where different notations are set and several geometric constructions and definitions used throughout the work are collected.

All manifolds are finite-dimensional, real, paracompact, connected and $C^\infty$. All
maps are $C^\infty$. Sum over crossed repeated indices is understood.


\section{The frame bundle and its canonical forms}
\label{framebundle}

\subsection{Basic definitions and notation}\label{sec:basic}

Consider a space-time manifold $M$ of dimension $m$. The corresponding bundle of frames (see for example \cite{KN1}) $\tau:LM\rightarrow M$ is the set\footnote{Alternatively, we can think of each of the fibers $LM_x$ as the set of ordered bases of the tangent space $T_xM$.}
\[
LM:=\bigcup_{x\in M}\mathop{\text{Iso}}{\left(\mR^m,T_xM\right)},
\]
where $\mathop{\text{Iso}}{\left(V,W\right)}$ is the set of (linear) isomorphisms between vector spaces $V$ and $W$.

It is well-known that the general linear group $G:=GL(m)$ acts naturally on $\mR^m$ by automorphisms. This action in turn induces a $G$-right action on $LM$, according to the formula
\[
u\cdot A:=u\circ A,
\]
for every $u\in LM,A\in G$, endowing $LM$ with a $G$-principal bundle structure.

Let us fix a matrix (any signature can be used in these considerations; the one chosen here follows closely the signature found in General Relativity)
\[
\eta:=\mathop{\text{diag}}{\left(-1,1,\cdots,1\right)};
\]
it can be considered either a map $\eta:\mR^m\rightarrow\left(\mR^m\right)^*$ or a bilinear form $\eta:\mR^m\times\mR^m\rightarrow\mR$.

Associated with the bundle $\tau$, we have the fiber bundle of $1$-jets $J^1\tau$ of sections of $\tau$. Given a section $s\in \Gamma(\tau)$, the 1-jet of $s$ at $x$, denoted $j^1_xs$, is the class of local sections being contact equivalent up to first order at $x$. This space has natural bundle projections $\tau_{10}:J^1\tau\rightarrow LM$ and $\tau_{1}:J^1\tau\rightarrow M$.

For every element $A\in G$, right translation $R_A:LM\rightarrow LM$ is a bundle isomorphism over the identity and so it can be lifted to a bundle isomorphism of $J^1\tau$ by taking the $1$-jet $j^1R_A:J^1\tau\rightarrow J^1\tau$. Accordingly, this defines a right action of $G$ on $J^1\tau$ and it can be checked that the quotient $C(LM):=J^1\tau/G$ is a smooth manifold, making $q:J^1\tau\rightarrow C(LM)$ into a $G$-principal bundle fitting the diagram
\[
\begin{diagram}
  \node[2]{J^1\tau}\arrow[1]{se,t}{\tau_{10}}\arrow[2]{s,l}{\tau_1}\arrow[1]{sw,l}{q}\\
\node{C(LM)}\arrow[1]{se,b}{p}  \node[2]{LM}\arrow{sw,b}{\tau}\\
  \node[2]{M}
\end{diagram}
\]
The bundle $C(LM)$ is called the {\sl bundle of connections of} $LM$ (see \cite{springerlink:10.1007/PL00004852} for an account of the geometry of this bundle). It can be proved that the bundles $J^1\tau\rightarrow C(LM)$ and $C(LM)\times_M LM\rightarrow C(LM)$ are diffeomorphic. It is worth mentioning that, under this identification, $q=pr_1$ and $\tau_{10}=pr_2$; where $pr_i$ is the projection on the $i$-th factor. Furthermore, the action of $G$ reduces to the action in the second factor, i.e. if $\rho\in J^1\tau$, $u=\tau_{10}(\rho)$, $[\rho]_G:=q(\rho)$ and $A\in G$, then $\rho\cdot A=([\rho]_G,u\cdot A)$.

If $\rho\in \tau_{10}^{-1}(u)$, then $\rho$ can also be thought of as a linear map $\rho:T_{\tau(u)}M\rightarrow T_{u}LM$ such that $T_u\tau \circ \rho=Id_{T_{\tau(u)}M}$. The interpretation goes as follows: given a local section $s\in \Gamma(\tau)$ and a tangent vector $X\in T_{x}M$, then $j^1_xs(X)=T_xs(X)$. Accordingly, each element $[\rho]_G\in C(LM)$ can be interpreted as a linear map $[\rho]_G:T_xM\rightarrow \left.(TLM)/G\right|_x$ such that $[T\tau]_G\circ [\rho]_G=Id_{T_{x}M}$, where the action of $G$ in $TLM$ is naturally induced by the action of $G$ on $LM$, and $[T\tau]_G:(TLM)/G\rightarrow TM$ is given by $[T\tau]([X])=T\tau(X)$.

Coordinates in $M$ will be denoted using greek indices $(x^\mu)$ and the related fiber coordinates in $LM$ will be denoted $(x^\mu,e^\nu_k)$, where $u(\mathbf{e}_k)=e^\nu_k\partial/\partial x^\nu$ and $\{\mathbf e_1,\dots,\mathbf e_m\}$ is the canonical basis in $\mathbb R^m$. Accordingly, fiber coordinates in $J^1\tau$ will be denoted $(x^\mu,e^\nu_k,e^\nu_{k\sigma})$. Using these coordinates, it can be seen that $(x^\mu,\Gamma^\mu_{\nu\sigma}:=-e^k_\nu e^\mu_{k\sigma})$ define fiber coordinates in $C(LM)$.

\subsection{The universal principal connection}

Now, we define a principal connection on the bundle $q:J^1\tau\rightarrow C(LM)$ fulfilling a universal property. First, observe that every element $[\rho]_G\in C(LM)$ can be viewed as a pointwise connection, i.e. every $[\rho]_G$ defines a unique family of projections $\Gamma_u:T_uLM\rightarrow V_u\tau$ for every $u\in \tau^{-1}(\tau(\rho))$. Indeed, if $\rho$ is any representative of the class $[\rho]_G$, set $u=\tau_{10}(\rho)$ and define 
\[
\Gamma_{u}:=Id_{T_{u}LM}-\rho\circ T_{u}\tau.
\]
It is immediate to see that $\Gamma_u(X)\in V_u\tau$, for every $X\in T_uLM$. For any other element $u'\in\tau^{-1}(\tau(u))$, we use right translation as follows
\[
\Gamma_{u'}:=T_{u}R_g\circ\Gamma_u\circ T_{u'}R_{g^{-1}},
\]
where $u'=u\cdot g$. It is readily seen that this construction is independent of the choice of the representative $\rho$.
\begin{note}
When treating principal connections, we will use the symbol $\Gamma$ to refer to the family of vertical projections. Furthermore, each principal connection carries a Lie algebra-valued differential form called connection form, which we denote by the symbol $\omega$. In the sequel, we refer to principal connections either through the projections $\Gamma$ or through its connection form $\omega$.
\end{note}
\begin{note}
Taking into account the previous observation, it is clear that the set of principal connections $\Gamma$ on $LM$ is in one-to-one correspondence with the sections of the bundle of connections. We have just seen the correspondence $[\rho]_G\mapsto \Gamma$. The inverse correspondence is given by $\Gamma\mapsto \sigma_\Gamma(x)=[hor^u(x)]_G$, where $hor^u(x):T_xM\rightarrow T_uLM$ is the horizontal lift related to $u\in \tau^{-1}(x)$ and $\Gamma$.
\end{note}

Now denote by $\mathfrak g$ the Lie algebra of $G$ and define $\omega\in \Omega^1(J^1\tau,\mathfrak g)$ as
\[
\left.\omega\right|_{\rho}(Y):=[\rho]_G\left(T_\rho\tau_{10}(Y)\right),
\]
where we are using the identification $LM\times \mathfrak g\rightarrow V\tau$. The fact that this Lie algebra-valued differential form is indeed the connection form of a principal connection can be easily checked.
 
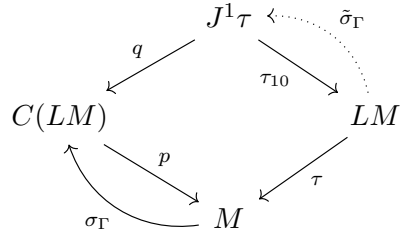
\[
\begin{tikzcd}
  &J^1\tau\ar[dr,"\tau_{10}"']\ar[dl,"q"']&\\
C(LM)\ar[dr,"p"] && LM\ar[dl,"\tau"]\ar[ul, bend right=40, dotted, "\tilde\sigma_\Gamma"']\\
  &M\ar[ul, bend left=40, "\sigma_\Gamma"]&
\end{tikzcd}
\]

To introduce the universal property associated with $\omega$, let us observe that, if $\Gamma$ is a principal connection on $LM$ and $\sigma_{\Gamma}$ is the related section of the bundle of connections, then we can define a section ${\tilde\sigma}_\Gamma\in\Gamma\tau_{10}$ by using the identification $J^1\tau\simeq C(LM)\times_M LM$ as (see the diagram above)
\[
{\tilde\sigma}_\Gamma(u)=(\sigma_\Gamma(\tau(u)),u).
\]
Then, if $\omega_\Gamma$ is the connection form associated with $\Gamma$, it turns out that $\omega_\Gamma=\tilde\sigma_\Gamma^*(\omega)$. In other words, any connection form of a principal connection can be recovered as a pullback of $\omega$ by the section $\tilde\sigma_\Gamma$. In this sense we say that $\omega$ is a universal connection. Accordingly, the universal curvature is given by (see Appendix \ref{sec:vvalform})
\[
\Omega:=\dif\omega+[\omega\wedcol\omega]
\]
and it can be seen that the curvature form associated with $\Gamma$ is $\Omega_\Gamma=\tilde\sigma_\Gamma^*(\Omega)$.

If $\{E^i_j\}$ denotes the canonical basis of $\mathfrak g$ and $\omega=\omega^i_jE^j_i$, then the coordinate expression of the forms $\omega^i_j$ using fiber coordinates is $\omega^i_j=e^i_\mu(\dif e^\mu_j-e^\mu_{j\sigma}\dif x^\sigma)$ and $\Omega^i_j=\dif\omega^i_j+\omega^i_k\wedge\omega^k_j$.

\subsection{The canonical form $\theta$}

In $LM$ we can define a canonical $\mathbb R^m$-valued 1-form $\bar\theta$ as follows. If $X\in T_uLM$, then
\[
\bar\theta(X)=u^{-1}\left(T_u\tau(X)\right).
\]
This allows us to define a similar form in $J^1\tau$, denoted $\theta$, as the pullback $\tau_{10}^*\bar\theta$.

In terms of the canonical basis of $\mR^m$, if we write $\theta=\theta^k\mathbf e_k$, it can be seen that the coordinate expression of the forms $\theta^k$ is given by $\theta^k=e^k_\mu\dif x^\mu$, where $e^k_\mu$ is such that $e^k_\mu e_j^\mu=\delta^k_j$ and $e^j_\nu e_j^\mu=\delta^\mu_\nu$. 

The form $\theta$ turns out to be a tensorial $1$-form of type $\text{Ad}$  (for details you can check \cite{KN1}). We can use the local expressions for $\theta^k$ and $\omega^i_j$ in a trivializing open set $U\subset J^1\tau$ to prove that these forms are linearly independent.

The exterior covariant derivative of $\theta$ with respect to $\omega$ gives rise to another differential form fulfilling a new universal property, called the universal torsion form $T$, i.e. (see Appendix \ref{sec:vvalform})
\[
T=\dif\theta+\omega\weddot\theta.
\]
The universal property in this case arise as follows: if $\Gamma$ is a principal connection on $LM$, then its related torsion form $T_\Gamma$ is recovered as the pullback
\[
T_\Gamma=\tilde\sigma_\Gamma^*(T).
\]
As we did before, we can express $T$ in terms of the canonical basis of $\mR^m$ by writing $T=T^k\mathbf e_k$ with
\[
T^k=\dif \theta^k+\omega^k_i\wedge\theta^i.
\]
A local expression for $T$ in a trivializing open set $U\subset J^1\tau$ can be obtained using those for $\omega$ and $\theta$. In fact
\begin{align*}
T^k=&\dif (e^k_\mu\dif x^\mu)+e^k_\mu(\dif e^\mu_i-e^\mu_{i\sigma}\dif x^\sigma)\wedge e^i_\nu \dif x^\nu\\
=&\dif e^k_\mu\wedge\dif x^\mu+e^i_\nu e^k_\mu \dif e^\mu_i\wedge \dif x^\nu-e^k_\mu e^\mu_{i\sigma}e^i_\nu \dif x^\sigma\wedge \dif x^\nu\\
=&\dif e^k_\mu\wedge\dif x^\mu-e^i_\nu   e^\mu_i \dif e^k_\mu\wedge \dif x^\nu+\frac{1}{2}e^k_\mu (e^\mu_{i\nu}e^i_\sigma-e^\mu_{i\sigma}e^i_\nu )\dif x^\sigma\wedge \dif x^\nu\\
=&\frac{1}{2}e^k_\mu (e^\mu_{i\nu}e^i_\sigma-e^\mu_{i\sigma}e^i_\nu )\dif x^\sigma\wedge \dif x^\nu.
\end{align*}

This last expression shows that on the set $e^\mu_{i\nu}e^i_\sigma-e^\mu_{i\sigma}e^i_\nu =0$ of each trivializing neighbourhood the torsion form $T$ vanishes identically. It turns out that all of these sets can be smoothly glued together and define a submanifold  $T_0\subset J^1\tau$, as the next proposition shows
\begin{prop}
There exists a submanifold $\iota_0:T_0\hookrightarrow J^1\tau$ such that $\iota_0^*T\equiv 0$.
\end{prop}
\begin{proof}
As we anticipated, the manifold $T_0$ is given locally by the conditions
\begin{equation}
e^\mu_{i\nu}e^i_\sigma-e^\mu_{i\sigma}e^i_\nu=0,\label{t0}
\end{equation}
for every $\mu,\nu,\sigma$. To see that this is independent of the choice of coordinates, consider another trivializing neighbourhood (having nonempty intersection with the first) with fibered coordinates $(\bar x^\mu,\bar e^\mu_k,\bar e^\mu_{k\sigma})$. Change of coordinates between these two sets is given by
\begin{align*}
\bar e^\sigma_k&=\frac{\partial \bar x^\sigma}{\partial x^\theta}e^\theta_k\\
\bar e^\mu_{k\nu}&=\left(\frac{\partial \bar x^\mu}{\partial x^\tau}e^\tau_{k\rho}+\frac{\partial^2 \bar x^\mu}{\partial x^\tau\partial x^\rho}e^\tau_k\right)\frac{\partial x^\rho}{\partial \bar x^\nu}
\end{align*}
so
\begin{align*}
\bar e^k_\sigma\bar e^\mu_{k\nu}-\bar e^k_\nu\bar e^\mu_{k\sigma}=&\frac{\partial  x^\theta}{\partial \bar x^\sigma}e^k_\theta\left(\frac{\partial \bar x^\mu}{\partial x^\tau}e^\tau_{k\rho}+\frac{\partial^2 \bar x^\mu}{\partial x^\tau\partial x^\rho}e^\tau_k\right)\frac{\partial x^\rho}{\partial \bar x^\nu}-\frac{\partial  x^\theta}{\partial \bar x^\nu}e^k_\theta\left(\frac{\partial \bar x^\mu}{\partial x^\tau}e^\tau_{k\rho}+\frac{\partial^2 \bar x^\mu}{\partial x^\tau\partial x^\rho}e^\tau_k\right)\frac{\partial x^\rho}{\partial \bar x^\sigma}\\
=&\frac{\partial  x^\theta}{\partial \bar x^\sigma}\frac{\partial \bar x^\mu}{\partial x^\tau}\frac{\partial x^\rho}{\partial \bar x^\nu}\left(e^k_\theta e^\tau_{k\rho}-e^k_\rho e^\tau_{k\theta}\right),
\end{align*}
which implies that the vanishing of the expression \eqref{t0} is independent of the particular trivializing set.
\end{proof}

\begin{note}
It is possible to prove that the action of $G$ preserves the manifold $T_0$, i.e. $T_0\cdot A\subset T_0$ for every $A\in G$. This allows us to define a $G$-action on $T_0$, making $T_0\rightarrow T_0/G$ into a principal $G$-bundle. Moreover, denoting $C_0(LM):=T_0/G$ and using the identification $J^1\tau\simeq C(LM)\times LM$, we get the identification $T_0\simeq C_0(LM)\times LM$. If we pullback the universal connection $\omega$ through $\iota_0$, we get a new universal property concernig torsionless connections, instead of arbitrary connections. We will use this fact to write the Griffiths variational principle for Lovelock gravity.
\end{note}

\subsubsection{The Sparling forms $\theta_{i_1\dots i_p}$}
For each $p\leq m$, define
\[
\theta_{i_1\dots i_p}=\frac{1}{(m-p)!}\varepsilon_{i_1\dots i_pi_{p+1}\dots i_m}\theta^{i_{p+1}}\wedge\dots\wedge\theta^{i_{m}}.
\]
It is readily seen that $\theta_{i_1\dots i_p}$ is completely antisymmetric in its indices. Additionaly we have:

\begin{lem}
Define $\sigma_0:=\theta^1\wedge\dots\wedge\theta^m$. Then 
\[
\theta_{i_1\dots i_p}=X_{i_p}\lrcorner\dots\lrcorner X_{i_1}\lrcorner\sigma_0,
\]
for any vector fields $X_{i_k}\in\mathfrak X(J^1\tau)$ projecting to $u(\mathbf e_{i_k})$, i.e. $T\tau_1(X_{i_k}(j^1_xu))=u(\mathbf e_{i_k})$.
\end{lem}
\begin{proof}
Let us proceed by induction on $p$. First, observe that 
\begin{equation}
X_i\lrcorner\theta^j=(u^{-1}\circ u)^j_i=\delta^j_i.\label{contraction}
\end{equation}
Then
\begin{align*}
X_i\lrcorner\sigma_0&=X_i\lrcorner\left(\theta^1\wedge\dots\wedge\theta^m\right)=\frac{1}{m!}X_i\lrcorner\left(\varepsilon_{i_1\dots i_m}\theta^{i_1}\wedge\dots\wedge\theta^{i_m}\right)\\
&=\frac{1}{m!}\sum_{k=1}^m(-1)^{k+1}\delta^{i_k}_i\varepsilon_{i_1\dots i_m}\theta^{i_1}\wedge\dots\wedge\widehat{\theta^{i_k}}\wedge\dots\wedge\theta^{i_m}\\
&=\frac{1}{m!}\sum_{k=1}^m(-1)^{k+1}\varepsilon_{i_1\dots i_{k-1}ii_{k+1}\dots i_m}\theta^{i_1}\wedge\dots\wedge\widehat{\theta^{i_k}}\wedge\dots\wedge\theta^{i_m},
\end{align*}
and renaming the indices
\begin{align*}
X_i\lrcorner\sigma_0&=\frac{1}{m!}\sum_{k=1}^m(-1)^{2k}\varepsilon_{ii_2\dots i_m}\theta^{i_2}\wedge\dots\wedge\theta^{i_m}\\
&=\frac{m}{m!}\varepsilon_{ii_2\dots i_m}\theta^{i_2}\wedge\dots\wedge\widehat{\theta^{i_k}}\wedge\dots\wedge\theta^{i_m}=\frac{1}{(m-1)!}\varepsilon_{ii_1\dots i_m}\theta^{i_1}\wedge\dots\wedge\widehat{\theta^{i_k}}\wedge\dots\wedge\theta^{i_m},
\end{align*}
which proves the case $p=1$. Inductively
\begin{align*}
X_{i_{p+1}}\lrcorner\theta_{i_1\dots i_p}&=X_{i_{p+1}}\lrcorner\left(\frac{1}{(m-p)!}\varepsilon_{i_1\dots  i_pj_{p+1}\dots j_m}\theta^{j_{p+1}}\wedge\dots\wedge\theta^{j_{m}}\right)\\
&=\sum_{k=1}^{m-p}\frac{(-1)^{k+1}}{(m-p)!}\varepsilon_{i_1\dots  i_pj_{p+1}\dots j_m}\delta^{j_{p+k}}_{i_{p+1}}\theta^{j_{p+1}}\wedge\dots\wedge\widehat{\theta^{j_{p+k}}}\wedge\dots\wedge\theta^{j_m}\\
&=\sum_{k=1}^{m-p}\frac{(-1)^{k+1}}{(m-p)!}\varepsilon_{i_1\dots  i_pj_{p+1}\dots j_{p+k-1}i_{p+1}j_{p+k+1}\dots j_m}\theta^{j_{p+1}}\wedge\dots\wedge\widehat{\theta^{j_{p+k}}}\wedge\dots\wedge\theta^{j_m},
\end{align*}
again renaming indices
\begin{align*}
X_{i_{p+1}}\lrcorner\theta_{i_1\dots i_p}&=\sum_{k=1}^{m-p}\frac{(-1)^{2k}}{(m-p)!}\varepsilon_{i_1\dots  i_{p+1}j_{p+2}\dots j_m}\theta^{j_{p+2}}\wedge\dots\wedge\theta^{j_m}\\
&=\frac{m-p}{(m-p)!}\varepsilon_{i_1\dots  i_{p+1}j_{p+2}\dots j_m}\theta^{j_{p+2}}\wedge\dots\wedge\theta^{j_m}=\theta_{i_1\dots i_{p+1}}.\qedhere
\end{align*}
\end{proof}

We use the Sparling forms to write down local expressions for the Lovelock Lagrangian and the equations of motion. To facilitate the related computations it is necessary to know some properties of these forms, so we collect some of them in the next proposition. In the proof (and in the rest of the paper) we use the properties of the Levi-Civita symbol and the generalized Kronecker delta listed in Appendix \ref{sec:levi}.

\begin{prop}\label{proptheta}
The following properties hold (the hat on an index indicates that this index has been suppressed):
\begin{enumerate}
\item for every $r\leq s$, $\theta^{i_1}\wedge\dots\wedge\theta^{i_r}\wedge\theta_{j_1\dots j_s}=\frac{(-1)^{r(s-r)}}{(s-r)!}\delta^{i_1\dots i_s}_{j_1\dots j_s}\theta_{i_{r+1}\dots i_{s}}$.
\item $\theta^k\wedge\theta_{i_1\dots i_p}=\sum_{r=1}^p(-1)^{p+r}\delta^k_{i_r}\theta_{i_1\dots\widehat{i_r} \dots i_p}$.
\item $\dif\theta_{i_1\dots i_p}=T^{l}\wedge\theta_{i_1\dots i_pl}+\sum_{r=1}^{p}(-1)^{p+r}\omega^{l}_{i_r}\wedge\theta_{i_1\dots \widehat{i_r}\dots i_pl}-\omega^{l}_{l}\wedge\theta_{i_1\dots i_p}$.
\end{enumerate}
\end{prop}
\begin{proof}
\begin{enumerate}
\item Let us compute
\begin{align*}
\theta^{i_1}\wedge\dots\wedge\theta^{i_r}\wedge\theta_{j_1\dots j_s}=&\frac{1}{(m-s)!}\theta^{i_1}\wedge\dots\wedge\theta^{i_r}\varepsilon_{j_1\dots j_s j_{s+1}\dots j_m}\theta^{j_{s+1}}\wedge\dots\wedge\theta^{j_m}\\
=&\frac{\varepsilon_{j_1\dots j_s j_{s+1}\dots j_m}}{(m-s)!(m-s+r)!}\delta^{i_1\dots i_rj_{s+1}\dots j_m}_{a_1\dots a_ra_{r+1}\dots a_{m-s+r}}\theta^{a_1}\wedge\dots\wedge\theta^{a_{m-s+r}}\\
=&\frac{\varepsilon_{j_1\dots j_s j_{s+1}\dots j_m}}{(m-s)!(m-s+r)!(s-r)!}\varepsilon^{i_{r+1}\dots i_{s}i_1\dots i_rj_{s+1}\dots j_m}\\
\phantom{=}&\varepsilon_{i_{r+1}\dots i_{s}a_1\dots a_{m-s+r}}\theta^{a_1}\wedge\dots\wedge\theta^{a_{m-s+r}}\\
=&\frac{\varepsilon_{j_1\dots j_s j_{s+1}\dots j_m}}{(m-s)!(s-r)!}\varepsilon^{i_{r+1}\dots i_{s}i_1\dots i_rj_{s+1}\dots j_m}\theta_{i_{r+1}\dots i_{s}}\\
=&\frac{1}{(s-r)!}\delta^{i_{r+1}\dots i_{s}i_1\dots i_r}_{j_1\dots j_s}\theta_{i_{r+1}\dots i_{s}}
=\frac{(-1)^{r(s-r)}}{(s-r)!}\delta^{i_1\dots i_{s}}_{j_1\dots j_s}\theta_{i_{r+1}\dots i_{s}}\ .
\end{align*}
\item From the first point of the proposition and taking $r=1$ and $s=p$,
\begin{align*}
\theta^k\wedge\theta_{i_1\dots i_p}&=\frac{(-1)^{p-1}}{(p-1)!}\delta^{ki_2\dots i_p}_{j_1\dots j_p}\theta_{i_2\dots i_p}=\frac{(-1)^{p-1}}{(p-1)!}\sum_{r=1}^p(-1)^{r+1}\delta^k_{j_r}\delta^{i_2\dots i_p}_{j_1\dots\hat{j_r}\dots j_p}\theta_{i_2\dots i_p}\\
&=\sum_{r=1}^p\frac{(-1)^{p+r}}{(p-1)!}\delta^k_{j_r}\delta^{i_2\dots i_p}_{j_1\dots\hat{j_r}\dots j_p}\theta_{i_2\dots i_p}=\sum_{r=1}^p\frac{(-1)^{p+r}}{(p-1)!}\delta^k_{j_r}(p-1)!\theta_{j_1\dots\hat{j_r}\dots j_p}\\
&=\sum_{r=1}^p(-1)^{p+r}\delta^k_{j_r}\theta_{j_1\dots\hat{j_r}\dots j_p} \ .
\end{align*}

We may also prove this by induction on $p$. The case $p=1$ is just Eq. \eqref{contraction}. Assuming that the formula holds for $p-1$, then
\begin{align*}
\theta^k\wedge\theta_{i_1\dots i_p}&=\theta^k\wedge \left(X_p\lrcorner\theta_{i_1\dots i_{p-1}}\right)\\
&=(X_{i_p}\lrcorner\theta^k)\wedge\theta_{i_1\dots i_{p-1}}-X_{i_p}\lrcorner\left(\theta^k\wedge\theta_{i_1\dots i_{p-1}}\right)\\
&=\delta^k_{i_p}\theta_{i_1\dots i_{p-1}}-X_{i_p}\lrcorner\left(\sum_{r'=1}^{p-1}(-1)^{p-1+r'}\delta^k_{i_{r'}}\theta_{i_1\dots \widehat{i_{r'}}\dots i_{p-1}}\right)\\
&=(-1)^{p+p}\delta^k_{i_p}\theta_{i_1\dots i_{p-1}}+\sum_{r'=1}^{p-1}(-1)^{p+r'}\delta^k_{i_{r'}}\theta_{i_1\dots \widehat{i_{r'}}\dots i_{p-1}i_p}\\
&=\sum_{r=1}^p(-1)^{p+r}\delta^k_{i_r}\theta_{i_1\dots\widehat{i_r} \dots i_p}\ .
\end{align*}
\item
Let us compute now the differential of $\theta_{i_1\dots i_p}$
\begin{align*}
\dif\theta_{i_1\dots i_p}&=\frac{1}{(m-p)!}\varepsilon_{i_1\dots i_pi_{p+1}\dots i_m}\dif\left(\theta^{i_{p+1}}\wedge\dots\wedge\theta^{i_m}\right)\\
&=\frac{1}{(m-p)!}\varepsilon_{i_1\dots i_pi_{p+1}\dots i_m}\sum_{k=1}^{m-p}(-1)^{k+1}\theta^{i_{p+1}}\wedge\dots\wedge\dif\theta^{i_{p+k}}\wedge\dots\wedge\theta^{i_m}\\
&=\sum_{k=1}^{m-p}\frac{(-1)^{2k}}{(m-p)!}\varepsilon_{i_1\dots i_pi_{p+k}i_{p+1}\dots\widehat{i_{p+k}}\dots i_m}\dif\theta^{i_{p+k}}\wedge\theta^{i_{p+1}}\wedge\dots\wedge\widehat{\theta^{i_{p+k}}}\wedge\dots\wedge\theta^{i_m}.
\end{align*}
Renaming the indices and using the first part of the proposition
\begin{align*}
\dif\theta_{i_1\dots i_p}&=\left(T^l-\omega^l_{k}\wedge\theta^k\right)\wedge\theta_{i_1\dots i_pl}\\
&=T^l\wedge\theta_{i_1\dots i_pl}-\omega^l_{k}\wedge\theta^k\wedge\theta_{i_1\dots i_pl}\\
&=T^{l}\wedge\theta_{i_1\dots i_pl}+\sum_{r=1}^{p}(-1)^{p+r}\omega^{l}_{i_r}\wedge\theta_{i_1\dots \widehat{i_r}\dots i_pl}-\omega^{l}_{l}\wedge\theta_{i_1\dots i_p} \ .\qedhere
\end{align*}
\end{enumerate}
\end{proof}

\section{Variational problem for Lovelock gravity}
\label{varLov}

Griffiths variational problems \cite{book:852048} are posed on triples $(\Lambda\xrightarrow{\pi} M,\lambda,\mathcal I)$, where $\Lambda\xrightarrow{\pi}  M$ is a fiber bundle over the space-time $M$, $\lambda$ is an $m$-form that is $\pi$-horizontal (referred to as the Lagrangian form) and an exterior differential system $\mathcal I\subset\Omega^\bullet(\Lambda)$ \cite{BCG} characterizing the admissible sections of the problem.

\begin{defc}
The variational problem associated with a variational triple $(\Lambda\xrightarrow{\pi} M,\lambda,\mathcal I)$ consists in finding the sections $\sigma:M\rightarrow \Lambda$ which are integrals for $\mathcal I$ and are extremals for the functional
\[
S[\sigma]:=\int_M\sigma^*\lambda.
\]
\end{defc}

Remember that $\sigma$ is integral for $\mathcal I$ if and only if $\sigma^*\alpha=0$ for every $\alpha\in \mathcal I$. In particular, this implies that the variations of $S$ must be performed in such a way that the transformed sections remain integrals of $\mathcal I$. Hence, we define
\begin{defc}
Let $\mathcal I\subset\Omega^\bullet(\Lambda)$ be an exterior differential system. A local vector field $X\in\mathfrak X(\Lambda)$ is an infinitesimal symmetry of $\mathcal I$ if and only if
\[
\cL_X\cI\subset\cI.
\]
The set infinitesimal symmetries of $\cI$ is denoted $\text{Symm}(\cI)$.
\end{defc}

With this definition, it can be proved that the solutions to the variational problem associated with the variational triple $(\Lambda\xrightarrow{\pi} M,\lambda,\mathcal I)$ are those sections $\sigma$ integral for $\cI$ such that
\[
\sigma^*(X\lrcorner\dif\lambda)=0\qquad\text{for every }X\in\mathfrak X^{V\pi}(\Lambda)\cap \text{Symm}(\cI)
\]
which are, in turn, the field equations for this problem; here $\mathfrak{X}^{V\pi}(\Lambda)$ indicates the set of vector fields on $\Lambda$ which are vertical respect to the projection $\pi$.

\begin{note}
It could be possible for an exterior differential system to have no infinitesimal symmetries; nevertheless, it will be proved in Section \ref{sec:infin-symm-mathc} that the exterior differential system we will use in the variational problem for Lovelock gravity (see Equation \eqref{eq:EDSForLovelock} below) possesses non trivial infinitesimal symmetries.
\end{note}

\begin{note}Here we are assuming that $M$ is a manifold without boundary. Also, in order for $S$ to be well-defined, the form $\sigma^*\lambda$ must be compactly supported. In the sequel, we will assume that all the integrals we work with exist.\end{note}

\subsection{The Lovelock Lagrangian}

Now we are ready to define a Griffiths variational problem for Lovelock gravity \cite{lovel}. To do that we have to define the corresponding triple introduced in the previous section. The bundle chosen is $\tau_1:T_0\rightarrow M$, where $M$ is the $m$-dimensional smooth manifold representing space-time. Here we are writing $\tau_1$ instead of $\left.\tau_1\right|_{T_0}$ only to simplify the notation (we will do the same with the pullbacks of $\omega$ and $\theta$ through $\iota_0$).

As the exterior differential system restricting the admissible sections we take (see Appendix \ref{sec:cartan})
\[
\cI_{\text{L}}:=\left<\omega_{\mathfrak p}\right>_{diff} \ .
\]
The subscript $diff$ indicates the smallest exterior differential system containing the form $\omega_{\mathfrak p}$.

Using the canonical basis of $\mR^m$, we can alternatively describe $\mathcal I_{\text{L}}$ as the exterior differential system generated as follows
\begin{equation}\label{eq:EDSForLovelock}
  \cI_{\text{L}}=\left<\eta^{ik}\omega^j_k+\eta^{jk}\omega^i_k\right>_{diff}.
\end{equation}
It is also useful to define the forms $\omega^{ij}:=\eta^{ik}\omega^j_k$, in terms of which $\cI_{\text{L}}=\left<\omega^{ij}+\omega^{ji}\right>_{diff}$. Then we look for sections $\sigma\in\Gamma(\tau_1)$ fulfilling the condition
\[
\sigma^*\omega_{\mathfrak p}=0.
\]
It is customary to refer to this condition as the \emph{metricity} condition.

\begin{note}
Using the identification $T_0\simeq C_0(LM)\times_M LM$, every (local) section $\sigma\in\Gamma(\tau_1)$ over $U\subset M$ that is integral for $\cI_{\text{L}}$ can be thought of as a couple of sections $\sigma_1:=q\circ \sigma$ and $\sigma_2:=\tau_{10}\circ \sigma$. As we saw in the previous section, if $\Gamma$ is the principal connection induced by $\sigma_1$ on $\tau$, then $\omega_\Gamma=\tilde\sigma_1^*\omega$. Hence, the metricity condition implies that $\omega_\Gamma$ is a torsionless (pseudo) metric connection.
\end{note}

Following the constructions of Appendix \ref{sec:vvalform} about vector-valued differential forms, for every $k\leq n$, we can define a $k$-form with values in $\Lambda^k\mR^n$ given by
\[
\theta^{(k)}:=\underbrace{\theta\wedge\dots\wedge\theta}_{k\text{ times}}.
\]
Hence
\[
\theta^{(k)}(X_1,\dots,X_k):=\theta(X_1)\wedge\dots\wedge\theta(X_k).
\]
Using the canonical basis of $\mR^m$ we can write
\[
\theta^{(k)}=\theta^{i_1}\wedge\dots\wedge\theta^{i_k}\otimes \mathbf{e}_{i_1}\wedge\dots\wedge\mathbf{e}_{i_k}.
\] 
Now we can take the Hodge star operator in the second factor (see Appendix \ref{sec:hodge}), namely
\begin{align*}
\star\left(\theta^{(k)}\right)=&\theta^{i_1}\wedge\dots\wedge\theta^{i_k}\otimes \frac{1}{(m-k)!}\eta_{i_1j_1}\dots\eta_{i_kj_k}\varepsilon^{j_1\dots j_kj_{k+1}\dots j_m}\mathbf{e}_{j_{k+1}}\wedge\dots\wedge\mathbf{e}_{j_m}\\
=&\theta^{i_1}\wedge\dots\wedge\theta^{i_k}\otimes \frac{1}{(m-k)!}\eta_{i_1j_1}\dots\eta_{i_kj_k}\varepsilon^{j_1\dots j_m}\delta^{l_{k+1}}_{j_{k+1}}\dots \delta^{l_{m}}_{j_{m}}\mathbf{e}_{l_{k+1}}\wedge\dots\wedge\mathbf{e}_{l_m}\\
=&\theta^{i_1}\wedge\dots\wedge\theta^{i_k}\otimes \frac{1}{(m-k)!}\underbrace{\eta_{i_1j_1}\dots\eta_{i_kj_k}\eta_{r_{k+1}j_{k+1}}\dots\eta_{r_mj_m}\varepsilon^{j_1\dots j_m}}_{\det(\eta)\varepsilon_{i_1\dots i_kr_{k+1}\dots r_m}}\\
\phantom{=}&\eta^{r_{k+1}l_{k+1}}\dots\eta^{r_ml_m}\mathbf{e}_{l_{k+1}}\wedge\dots\wedge\mathbf{e}_{l_m},
\end{align*}
where we have used the properties of the Levi-Civita symbol (see Appendix \ref{sec:levi}). Now, renaming indices and using the definition of the forms $\theta_{i_1\dots i_p}$, we find
\begin{align}
\star\left(\theta^{(k)}\right)&=\det(\eta)\eta^{i_1j_1}\dots\eta^{i_pj_p}\theta_{i_1\dots i_p}\otimes\mathbf{e}_{j_1}\wedge\dots\wedge\mathbf{e}_{j_p},\label{starteta}
\end{align}
with $p=m-k$.
Notice that $\theta^{(k)}$ is a $k$-form with values in $\Lambda^{k}\mR^m$, while $\star\left(\theta^{(k)}\right)$ is a $k$-form with values in $\Lambda^{m-k}\mR^m$.
Now let $\displaystyle r< \left[\frac{m}{2}\right]$ be an integer, where 
$[\cdot]$ denotes the integral part. 

\begin{defc}
Let $V$ be an $m$-dimensional real vector space. We define
\[
A_r:\Lambda^{2r}V\rightarrow (\Lambda^r V)\otimes (\Lambda^r V)^*\simeq \left(\text{End}(\Lambda^rV)\right)^*
\]
as the unique linear map whose action on elementary $2r$-vectors is given by
\medskip
\begin{align*}
\left(v_{j_1}\wedge\dots\wedge v_{j_r}\right)\wedge&\left(v_{j_{r+1}}\wedge\dots\wedge v_{j_{2r}}\right)\ \mapsto \\
\mapsto&\frac{1}{(2r!)}\sum_{\sigma\in S_{2r}}\text{sgn}(\sigma)\left(v_{j_{\sigma(1)}}\wedge\dots\wedge v_{j_{\sigma(r)}}\right)\otimes\hat\eta^\flat\left(v_{j_{\sigma(r+1)}}\wedge\dots\wedge v_{j_{\sigma(2r)}}\right) \ ,
\end{align*}
\medskip
where $\hat\eta$ is the extension of $\eta$ to $\Lambda^rV$ defined on Appendix \ref{sec:hodge}. It is readily seen that it is in fact well-defined and linear.
\end{defc}

Using the linear map $A_r$, we can construct an $(m-2r)$-form with values in $\left(\text{End}(\Lambda^r\mR^m)\right)^*$ as 
\[
\Xi_r=A_r\left(\star\left(\theta^{(m-2r)}\right)\right).
\]
We can think of $\Xi_r$ as taking values in $(\Lambda^r\mathfrak g)^*$ rather than $\left(\text{End}(\Lambda^r\mR^m)\right)^*$ because the latter can be viewed as a subspace of the former. That is
\begin{equation}\label{eq:WedgeStarInStarWedge}
  \left(\Lambda^r\g\right)^*\supset\left(\text{End}\left(\Lambda^r\mR^m\right)\right)^*,
\end{equation}
which is a consequence of the inclusion $\Lambda^r\mathfrak g\subset\text{End}(\Lambda^r\mR^m)$ given by the monomorphism $\Gamma:\Lambda^r\left(\text{End}(\mR^m)\right)\to\text{End}\left(\Lambda^r\mR^m\right)$ defined as
\[
 A_1\wedge\cdots\wedge A_r\mapsto\left[v_1\wedge\cdots\wedge v_r\mapsto\frac{1}{r!}\sum_{\sigma\in S_r}A_1\left(v_{\sigma\left(1\right)}\right)\wedge\cdots\wedge A_r\left(v_{\sigma\left(r\right)}\right)\right].
\]
We can use these considerations to introduce the Lovelock Lagrangian:

\begin{defc}
The Lovelock Lagrangian is the $\tau_1$-horizontal $m$-form
\[
\lambda_{\text{L}}:=\sum_{r<[m/2]}a_r\left<\Xi_r\wedcol \Omega^r\right>,
\]
where $a_r$ are constants and
$\Omega^r=\underbrace{\Omega\wedge\dots\wedge\Omega}_{r\text{ times}}$.
\end{defc}

{
  \begin{note}[Regularity of the Lovelock Lagrangian]
    From the viewpoint of classical second order field theory on the bundle of metrics, the Lovelock Lagrangian is a singular Lagrangian: It follows from the fact that the equations of motion are given by the Einstein tensor, which are of second order, although for a regular Lagrangian in a second order theory, they should have been of fourth order.
    
    Nevertheless, it should be stressed that the variational problem posed by $\lambda_{\text{L}}$ and the constraints $\cI_{\text{L}}$ is not a classical one; in this regard, the extremals of such problem are not necessarily holonomic as sections of the jet bundle $J^1\tau$. In particular, the notion of regularity of a Lagrangian has not a clear generalization to this case; in fact, it would depend on which feature of this concept we want to highlight:
    \begin{enumerate}
    \item For example, the regularity of a Lagrangian can be seen as a sufficient condition for the existence of solutions for the equations of motion (as it allows us to apply Cauchy-Kovalevskaya theorem). From this viewpoint, the fact that the exterior differential system \eqref{eq:EDSForLovelock} representing these equations of motion admits solutions becomes a necessary condition for the regularity of the variational problem determined by the data $\left(\lambda_{\text{L}},\cI_{\text{L}}\right)$.
    \item For classical variational problems, it could be noted that the regularity of a Lagrangian is also tied to the fact that the associated Legendre transformation is a diffeomorphism. It suggests another way to generalize regularity for a variational problem of the type discussed here. The idea is that, from the unified formalism perspective, Legendre transformation becomes part of the equations of motion, and it can be obtained as a consequence of the direct sum structure of the multimomentum bundle (that in our case is determined by Lemma \ref{lem:MultiDecomposition} below). Using the equations of motion \eqref{eq:EDSForLovelock}, it results that the map generalizing Legendre transform in this sense is identically zero.
    \end{enumerate}
    In any case, this generalization would require further research, which is expected to be carried out elsewhere.
\end{note}
}

\subsubsection{Expressions in terms of the canonical basis of $\mR^m$}

If we denote by $\{\mathbf e^1,\dots,\mathbf e^m\}$ the dual basis of the canonical basis in $\mR^m$, we can write
\[
A_r(\mathbf e_{j_1}\wedge\dots\wedge \mathbf e_{j_{2r}})=\frac{1}{(2r)!}\sum_{\sigma\in S_{2r}}\text{sgn}(\sigma)\eta_{j_{\sigma(1)}l_1}\dots\eta_{j_{\sigma(r)}l_r}\mathbf e_{j_{\sigma(r+1)}}\wedge\dots\wedge \mathbf e_{j_{\sigma(2r)}}\otimes\mathbf e^{l_1}\wedge\dots\wedge\mathbf e^{l_r}.
\]

Also, using \eqref{starteta},
\begin{align*}
\Xi_r=&A_r\left(\star\left(\theta^{(m-2r)}\right)\right)\\
=&\det(\eta)\eta^{i_1j_1}\dots\eta^{i_{2r}j_{2r}}\theta_{i_1\dots i_{2r}}\otimes A_r\left(\mathbf{e}_{j_1}\wedge\dots\wedge\mathbf{e}_{j_{2r}}\right)\\
=&\frac{\det(\eta)}{(2r)!}\sum_{\sigma\in S_{2r}}\text{sgn}(\sigma)\eta^{i_1j_1}\dots\eta^{i_{2r}j_{2r}}\eta_{j_{\sigma(1)}l_1}\dots\eta_{j_{\sigma(r)}l_r}\theta_{i_1\dots i_{2r}}\otimes \mathbf e_{j_{\sigma(r+1)}}\wedge\dots\\
\phantom{=}&\dots\wedge \mathbf e_{j_{\sigma(2r)}}\otimes\mathbf e^{l_1}\wedge\dots\wedge\mathbf e^{l_r},
\end{align*}
but since $\theta_{i_1\dots i_{2r}}=\text{sgn}(\sigma)\theta_{i_{\sigma(1)}\dots i_{\sigma(2r)}}$,
\begin{align*}
\Xi_r=&\frac{\det(\eta)}{(2r)!}\sum_{\sigma\in S_{2r}}\eta^{i_{\sigma(1)}j_{\sigma(1)}}\dots\eta^{i_{\sigma(2r)}j_{\sigma(2r)}}\eta_{j_{\sigma(1)}l_1}\dots\eta_{j_{\sigma(r)}l_r}\theta_{i_{\sigma(1)}\dots i_{\sigma(2r)}}\otimes \mathbf e_{j_{\sigma(r+1)}}\wedge\dots\\
\phantom{=}&\dots\wedge \mathbf e_{j_{\sigma(2r)}}\otimes\mathbf e^{l_1}\wedge\dots\wedge\mathbf e^{l_r}\\
=&\frac{\det(\eta)}{(2r)!}\sum_{\sigma\in S_{2r}}\eta^{i_{\sigma(1)}j_{\sigma(1)}}\dots\eta^{i_{\sigma(2r)}j_{\sigma(2r)}}\eta_{j_{\sigma(1)}l_1}\dots\eta_{j_{\sigma(r)}l_r}\theta_{i_{\sigma(1)}\dots i_{\sigma(2r)}}\otimes \mathbf e_{j_{\sigma(r+1)}}\wedge\dots\\
\phantom{=}&\dots\wedge \mathbf e_{j_{\sigma(2r)}}\otimes\mathbf e^{l_1}\wedge\dots\wedge\mathbf e^{l_r}\\
=&\frac{\det(\eta)}{(2r)!}\sum_{\sigma\in S_{2r}}\eta^{i_{\sigma(r+1)}j_{\sigma(r+1)}}\dots\eta^{i_{\sigma(2r)}j_{\sigma(2r)}}\theta_{i_{\sigma(1)}\dots i_{\sigma(2r)}}\otimes \mathbf e_{j_{\sigma(r+1)}}\wedge\dots\\
\phantom{=}&\dots\wedge \mathbf e_{j_{\sigma(2r)}}\otimes\mathbf e^{i_{\sigma(1)}}\wedge\dots\wedge\mathbf e^{i_{\sigma(r)}} \ ,
\end{align*}
and, as all are dummy indices, we finally get
\[
\Xi_r=\det(\eta)\eta^{i_{r+1}j_{r+1}}\dots\eta^{i_{2r}j_{2r}}\theta_{i_{1}\dots i_{2r}}\otimes \mathbf e_{j_{r+1}}\wedge\dots\wedge \mathbf e_{j_{2r}}\otimes\mathbf e^{i_{1}}\wedge\dots\wedge\mathbf e^{i_{r}}.
\]

Furthermore, if $\Omega=\Omega^a_b\otimes\mathbf e^b\otimes \mathbf e_a$, we have
\[
\Omega^r=\Omega^{a_1}_{b_1}\wedge\dots\wedge \Omega^{a_r}_{b_r}\otimes\left(\mathbf e^{b_1}\otimes \mathbf{e}_{a_1}\right)\wedge\dots\wedge\left(\mathbf e^{b_r}\otimes\mathbf{e}_{a_r}\right)\in\Omega^{2r}\left(J^1\tau\right)\otimes\Lambda^r\g,
\]
so that, in view of the inclusion \eqref{eq:WedgeStarInStarWedge}, we obtain
\[
\lambda_{\text{L}}=\sum_{r<[m/2]}a_r\theta_{i_{1}\dots i_rl_1\dots l_{r}}\wedge\Omega^{i_1l_1}\wedge\dots\wedge \Omega^{i_rl_r},
\]
where $\Omega^{ab}=\eta^{bq}\Omega^{a}_{q}$ and all the possible multiplicative constants have been absorbed in the constants $a_r$. From now on, we will work with each homogeneous component
\[
\lambda^{(r)}_{\text{L}}=\theta_{i_{1}\dots i_rl_1\dots l_{r}}\wedge\Omega^{i_1l_1}\wedge\dots\wedge \Omega^{i_rl_r},
\]
which will be indicated with the generic symbol $\lambda_L$.
\begin{note}
To simplify the computations, it will be convenient to introduce the following multi-index notation. We use capital letters $I,J$ to denote multi-indices $I=(i_1,\dots,i_p)$, $J=(j_1,\dots,j_p)$. An apostrophe denotes a multi-index formed by removing the first index of a given multi-index, i.e. $I'=(i_2,\dots,i_p)$ if $I=(i_1,\dots,i_p)$. In this case, we use concatenation of indices and multi-indices and write
$I=i_1I'$.
Also, we will write $\Omega^{IJ}=\Omega^{i_1j_1}\wedge\dots\wedge \Omega^{i_rj_r}$ and $\theta_{i_{1}\dots i_rj_1\dots j_{r}}=\theta_{IJ}$. Thus, the Lovelock Lagrangian can be written
\[
\lambda^{(r)}_{\text{L}}=\theta_{IJ}\wedge\Omega^{IJ}.
\]

\end{note}

\subsubsection{Relation with the metric-affine Lagrangian}

To relate $\lambda_{\text{L}}$ with the metric-affine formalism, remember that every principal connection $\Gamma$ gives rise to a linear connection in $TM$ (as an associated vector bundle with fiber $\mR^m$ and canonical action of $G$). Let $\omega_\Gamma$ is the corresponding connection form (obtained as the pullback of the universal connection $\omega$ by a suitable section) and $\Omega_\Gamma$ its related curvature. 

Then, if $\Omega_\Gamma^{ab}=\Omega^{ab}_{\mu\nu}\dif x^\mu\wedge\dif x^\nu$, we have
\begin{equation}
\Omega^{ab}_{\mu\nu}=R_{\mu\nu}^{\sigma\tau}e^a_\sigma e^b_\tau,\label{omegar}
\end{equation}
where $R_{\mu\nu}^{\sigma\tau}=g^{\rho\tau}R_{\rho\mu\nu}^{\sigma}$ are the components of the curvature tensor with respect to the linear connection induced by $\Gamma$, i.e.
\[
R\left(\frac{\partial}{\partial x^\mu},\frac{\partial}{\partial x^\nu}\right)\frac{\partial}{\partial x^\rho}=R_{\rho\mu\nu}^{\sigma}\frac{\partial}{\partial x^\sigma} \ ,
\]
and $g_{\mu\nu}:=e_\mu^a\eta_{ab}e^b_{\nu}$ is the corresponding metric.

Thus, we can compute locally the pullback of $\lambda_{\text{L}}$ by a section as follows
\begin{align*}
\lambda_{\text{L}}=&\theta_{i_{1}\dots i_rl_1\dots l_{r}}\wedge\Omega^{i_1l_1}\wedge\dots\wedge \Omega^{i_rl_r}\\
=&\varepsilon_{i_{1}\dots i_rl_1\dots l_{r}s_1\dots s_k}\theta^{s_1}\wedge\dots\wedge\theta^{s_k}\wedge\Omega^{i_1l_1}\wedge\dots\wedge \Omega^{i_rl_r}\\
=&\varepsilon_{i_{1}\dots i_rl_1\dots l_{r}s_1\dots s_k}e^{s_1}_{\rho_1}\dots e^{s_k}_{\rho_k}\Omega^{i_1l_1}_{\mu_1\nu_1}\dots\Omega^{i_rl_r}_{\mu_r\nu_r}\dif x^{\rho_1}\wedge\dots\wedge\dif x^{\rho_k}\wedge\dif x^{\mu_1}\wedge\dif x^{\nu_1}\wedge\dots\\
\phantom{=}&\dots\wedge\dif x^{\mu_r}\wedge\dif x^{\nu_r}\\
=&\varepsilon_{i_{1}\dots i_rl_1\dots l_{r}s_1\dots s_k}\varepsilon^{\rho_{1}\dots \rho_k\mu_1\nu_1\dots \mu_{r}\nu_r}e^{s_1}_{\rho_1}\dots e^{s_k}_{\rho_k}\Omega^{i_1l_1}_{\mu_1\nu_1}\dots\Omega^{i_rl_r}_{\mu_r\nu_r}\dif^m x\\
=&\varepsilon_{i_{1}\dots i_rl_1\dots l_{r}s_1\dots s_k}\varepsilon^{\rho_{1}\dots \rho_k\mu_1\nu_1\dots \mu_{r}\nu_r}e^{s_1}_{\rho_1}\dots e^{s_k}_{\rho_k}\left(\delta^{i_1}_{c_1}\delta^{l_1}_{d_1}\right)\dots\relax\left(\delta^{i_r}_{c_r}\delta^{l_r}_{d_r}\right)\Omega^{c_1d_1}_{\mu_1\nu_1}\dots\Omega^{c_rd_r}_{\mu_r\nu_r}\dif^m x\\
=&\varepsilon_{i_{1}\dots i_rl_1\dots l_{r}s_1\dots s_k}\varepsilon^{\rho_{1}\dots \rho_k\mu_1\nu_1\dots \mu_{r}\nu_r}e^{s_1}_{\rho_1}\dots e^{s_k}_{\rho_k}\left(e^{i_1}_{\alpha_1}e^{\alpha_1}_{c_1}e^{l_1}_{\beta_1}e^{\beta_1}_{d_1}\right)\dots\\
\phantom{=}&\dots\relax\left(e^{i_r}_{\alpha_r}e^{\alpha_r}_{c_r}e^{l_r}_{\beta_r}e^{\beta_r}_{d_r}\right)\Omega^{c_1d_1}_{\mu_1\nu_1}\dots\Omega^{c_rd_r}_{\mu_r\nu_r}\dif^m x\\
=&\det(e)\varepsilon^{\rho_{1}\dots \rho_k\mu_1\nu_1\dots \mu_{r}\nu_r}\varepsilon_{\alpha_{1}\dots \alpha_r\beta_1\dots \beta_{r}\rho_1\dots \rho_k} R^{\alpha_1\beta_1}_{\mu_1\nu_1}\dots R^{\alpha_r\beta_r}_{\mu_r\nu_r}\dif^m x\\
=&(m-2r)!\det(e)\delta^{\mu_1\nu_1\dots \mu_{r}\nu_r}_{\alpha_{1}\dots \alpha_r\beta_1\dots \beta_{r}}R^{\alpha_1\beta_1}_{\mu_1\nu_1}\dots R^{\alpha_r\beta_r}_{\mu_r\nu_r}\dif^m x\\
=&(-1)^{r!}(m-2r)!\det(e)\delta^{\mu_1\nu_1\dots \mu_{r}\nu_r}_{\alpha_{1}\beta_1\dots \alpha_r\beta_r}R^{\alpha_1\beta_1}_{\mu_1\nu_1}\dots R^{\alpha_r\beta_r}_{\mu_r\nu_r}\dif^m x,
\end{align*}
where we have used Eq. \eqref{omegar} and the identity
\[
\varepsilon_{i_{1}\dots i_rl_1\dots l_{r}s_1\dots s_k}e^{i_1}_{\alpha_1}\dots e^{i_r}_{\alpha_r}e^{l_1}_{\beta_1}\dots e^{l_r}_{\beta_r}e^{s_1}_{\rho_1}\dots e^{s_k}_{\rho_k}=\varepsilon_{\alpha_{1}\dots \alpha_r\beta_1\dots \beta_{r}\rho_1\dots \rho_k}\det(e).
\]
Then, as $\det(e)=\sqrt{-g}$, we recover the usual Lovelock Lagrangian, i.e.
\[
\lambda_{\text{L}}=\alpha\sqrt{-g}\delta^{\mu_1\nu_1\dots \mu_{r}\nu_r}_{\alpha_{1}\beta_1\dots \alpha_r\beta_r}R^{\alpha_1\beta_1}_{\mu_1\nu_1}\dots R^{\alpha_r\beta_r}_{\mu_r\nu_r}\dif^m x.
\]

\section{Field equations}
\label{fieqs}

As we have said in the previous section, to compute the field equations associated with the Lovelock problem we need to characterize the infinitesimal symmetries of the exterior differential system $\mathcal I_{\text{L}}$. We devote the following section to this task.

\subsection{Infinitesimal symmetries of $\mathcal I_{\text{L}}$}
\label{sec:infin-symm-mathc}
To give a characterization of the infinitesimal symmetries of $\mathcal I_{\text{L}}$, it is useful to introduce a convenient basis of vector fields for the vertical bundle $V\tau_1$. Using the identification $J^1\tau\simeq C(LM)\times_M LM$ we can construct this basis using $q$-vertical and $\tau_{10}$-vertical vector fields. 

First, consider the infinitesimal generators associated with the action of $G$ on $J^1\tau$. If $E^i_j$ denote a vector of the canonical basis of $\mathfrak g$, we denote these vector fields by $(E^i_j)_{J^1\tau}$. It can be seen that $T\tau_{10}(E^i_j)_{J^1\tau}=(E^i_j)_{LM}$ and hence they are $\tau_1$-vertical vector fields. From the principal bundle structure of $q:J^1\tau\rightarrow C(LM)$ it is immediate to see that they are also $q$-vertical.

Furthermore, as $J^1\tau\rightarrow LM$ is an affine bundle, we can construct vertical lifts of $\tau$-vertical vector fields. Given a differential form $\alpha\in\Omega(M)$ and a $\tau$-vertical vector field $X$, the vertical lift $(\alpha,X)^V$ is defined as the vector field whose flow is given by 
\[
\theta(t,j^1_xs)=j^1_xs+t\alpha_x\otimes X(s(x)),
\]
where the sign ``$+$'' must be understood as the affine action of $\tau^*\pi_{M}\otimes_{LM}V\tau$ on $J^1\tau$.\footnote{Here $\pi_{M}:T^*M\rightarrow M$ is the cotangent projection.} In other words, 
\[
(\alpha,X)^V(j^1_xs)=\left.\vec{\frac{d}{dt}}\right|_{t=0}\left(j^1_xs+t\alpha_x\otimes X(s(x))\right),
\]

We can adapt this definition replacing the differential forms $\alpha$ by differential forms along $J^1\tau$, i.e. $\alpha\in\Gamma(\tau_1^*\pi_M)$,
\[
(\alpha,X)^V(j^1_xs)=\left.\vec{\frac{d}{dt}}\right|_{t=0}\left(j^1_xs+t\left.\alpha\right|_{j^1_xs}\otimes X(s(x))\right),
\]
In particular, we can use the forms $\theta^r$ and the infinitesimal generators $(E^s_t)_{LM}$, which we denote $(\theta^r,(E^s_t)_{LM})^V$. It is clear that these vector fields are $\tau_{10}$-vertical, and in consequence they are also $\tau_1$-vertical.

It can be proved that the set of vector fields $\left\{(E^s_t)_{J^1\tau},(\theta^r,(E^s_t)_{LM})^V\right\}$ form a basis of the vertical bundle $V\tau_1$ (see \cite{doi:10.1142/S0219887818500445}).

\begin{note}
It is useful to write down local expressions for the vector fields introduced above. In a trivializing open set, it can be seen that
\[
(E^k_l)_{J^1\tau}(x^\mu,e^\mu_i,e^\mu_{i\sigma})=e^\mu_l\frac{\partial}{\partial e^\mu_k}+e^\mu_{l\sigma}\frac{\partial}{\partial e^\mu_{k\sigma}} \ ,
\]
and
\[
(\theta^r,(E^s_t)_{LM})^V(x^\mu,e^\mu_i,e^\mu_{i\sigma})=e^r_\sigma e^\mu_t e^s_\nu e^\nu_i\frac{\partial}{\partial e^\mu_{i\sigma}}\ .
\]
Using these expressions we can check that the vector fields $\left\{(E^s_t)_{J^1\tau},(\theta^r,(E^s_t)_{LM})^{V_0}\right\}$, with
\[
(\theta^r,(E^s_t)_{LM})^{V_0}=(\theta^r,(E^s_t)_{LM})^V+(\theta^s,(E^r_t)_{LM})^V \ ,
\]
form a basis of the bundle $V(\left.\tau_1\right|_{T_0})$.
\end{note}

\begin{note}
If we fix a principal connection (that may be chosen torsionless) on $\tau$, it is possible to complete this basis to a full basis of $TJ^1\tau$ by considering the prolongations of the standard horizontal vector fields on $LM$ (see \cite{KN1}).
\end{note}

\begin{prop}\label{contractions}
The following contractions hold
\begin{align*}
(\theta^r,(E^s_t)_{LM})^V\lrcorner \Omega^k_l&=\delta^k_t\delta^s_l\theta^r \ ,\\
(\theta^r,(E^s_t)_{LM})^V\lrcorner \omega^k_l&=0\ ,\\
(\theta^r,(E^s_t)_{LM})^V\lrcorner \theta^k&=0\ .
\end{align*}
\end{prop}

Before moving on, remember that given a vector field $U\in\mathfrak X(LM)$ the first prolongation of $U$ is the unique vector field $j^1U\in\mathfrak X(J^1\tau)$ that is projectable to $U$ and is an infinitesimal symmetry of the contact exterior differential system. The next lemma shows that prolongations of $G$-invariant vector fields are infinitesimal symmetries of the universal connection {\cite{springerlink:10.1007/PL00004852}:

\begin{lem}
  Let $U\in\mathfrak{X}^{V\tau}\left(LM\right)$ be a vertical $G$-invariant vector field on $LM$. Then
  \[
    \cL_{j^1U}\omega=0.
  \]
\end{lem}
\begin{proof}
  We know that $U$ is $G$-invariant if, and only if, its flow $\Psi_t^U:LM\rightarrow LM$ is an automorphism of $LM$. Furthermore, for every automorphism $F\colon LM\rightarrow LM$, we have that
  \[
    \left(j^1F\right)^*\omega=\omega \ ,
  \]
and the lemma follows from this fact.
\end{proof}

\begin{note}\label{rem:VertVectorProlong}
Given any element $u\in LM$, there exists a neighborhood $V$ containing $u$ and a set of $G$-invariant vector fields $\left\{U_j^i\right\}$  generating $\mathfrak{X}^{V\tau}\left(V\right)$ as a $C^\infty\left(V\right)$-module. 
\end{note}

Now we show how to construct infinitesimal symmetries starting from $G$-invariant vertical vector fields

\begin{lem}\label{lem:InfSymmetryLift}
Let $\left\{f^l_j\right\}$ be a family of arbitrary functions on $\tau\left(V\right)$ and let $\left\{U^i_j\right\}$ be a basis of $G$-invariant local vector fields generating $\mathfrak X^{V\tau}(V)$. Then there exists a (non unique) family of functions $\left\{F^i_{kl}\right\}$ on $\tau^{-1}_{10}\left(V\right)\subset T_0$ such that
  \[
    Z:=f^l_jj^1U^j_l+F^i_{kl}\left(\theta^k,\left(E^l_i\right)_{LM}\right)^V
  \]
  is an infinitesimal symmetry of $\cI_{\text{L}}$ tangent to $T_0$.
\end{lem}
\begin{proof}
Let $\left\{U_j^i\right\}$ be the basis of $G$-invariant local vector fields generating $\mathfrak{X}^{V\tau}\left(V\right)$. Since the set of infinitesimal generators $\left\{(E^k_l)_{LM}\right\}$ form another basis, there exist smooth functions $M^{il}_{jk},N^{il}_{jk}\in C^\infty(V)$ such that
\[
U_j^i=M^{il}_{jk}(E^k_l)_{LM}\qquad\text{and}\qquad M^{ip}_{jq}N_{pk}^{ql}=\delta^i_k\delta^l_j.
\]
Now, from the formula
\[
  j^1\left(fW\right)=fj^1W+\left(Df,W\right)^V,\qquad f\in C^\infty\left(LM\right),W\in\mathfrak{X}^{V\tau}\left(LM\right),
\]
we obtain
\begin{equation}
  j^1U_j^i=M_{jk}^{il}\left(E_l^k\right)_{J^1\tau}+\left(D_rM_{jk}^{il}\right)\left(\theta^r,\left(E_l^k\right)_{LM}\right)^V,\label{jotaU}
\end{equation}
where $\displaystyle{D_rM_{jk}^{il}=\frac{\partial M_{jk}^{il}}{\partial x^\mu}e^\mu_r}$. In consequence, in order for $Z$ to be tangent to $T_0$, we must take
\begin{equation}
F_{kj}^i=-f^s_tD_kM_{sj}^{ti}+G^i_{kj}\label{Efes}
\end{equation}
with the functions $G^i_{kj}$ fulfilling $G^i_{kj}=G^i_{jk}$.

To compute the Lie derivative, let us write
\[
\cL_{Z}\omega_{\pf}^{pq}=\cL_{f^j_ij^1U^i_j}\omega_{\pf}^{pq}+\cL_{F^i_{kj}\left(\theta^k,\left(E^j_i\right)_{LM}\right)^V}\omega_{\pf}^{pq}
\]
and compute separately.

First, 
\begin{align*}
\cL_{f^j_ij^1U^i_j}\left(\omega_{\pf}\right)^{pq}&=f^j_i\cL_{j^1U^i_j}\left(\omega_{\pf}\right)^{pq}+\left(j^1U^i_j\lrcorner\left(\omega_{\pf}\right)^{pq}\right)\dif f^j_i
=\mu_j^{ipq}D_kf_i^j\theta^k,
  \end{align*}
  where $\mu_j^{ipq}:=j^1U^i_j\lrcorner\left(\omega_{\pf}\right)^{pq}$ and $\displaystyle{D_kf_i^j=\frac{\partial f_i^j}{\partial x^\mu}e^\mu_k}$. 

On the other hand
  \begin{align*}
\cL_{F_{kj}^i\left(\theta^k,\left(E_i^j\right)_{LM}\right)^{V}}\omega_{\pf}=F_{kj}^i\left(\theta^k,\left(E_i^j\right)_{LM}\right)^{V}\lrcorner \dif \omega_{\pf}
 =F_{kj}^i\left(\theta^k,\left(E_i^j\right)_{LM}\right)^{V}\lrcorner\Omega_{\pf},
  \end{align*}
  where $\left(\Omega_{\pf}\right)_q^p=\frac{1}{2}\left(\Omega_q^p+\eta_{qa}\Omega_b^a\eta^{bp}\right)$. Using Proposition \ref{contractions}
  \begin{align*}
    \left(\theta^k,\left(E_i^j\right)_{LM}\right)^V\lrcorner\left(\Omega_q^p+\eta_{qa}\Omega_b^a\eta^{bp}\right)&=\left(\delta^p_i\delta_q^j+\eta_{qa}\delta^a_i\delta^j_b\eta^{bp}\right)\theta^k
    =\left(\delta^p_i\delta_q^j+\eta_{qi}\eta^{jp}\right)\theta^k,
  \end{align*}
from which we deduce [recall \eqref{Efes}]
\begin{align*}
F_{kj}^i\left(\theta^k,\left(E_i^j\right)_{LM}\right)^{V}\lrcorner\Omega_{\pf}^{pq}&=-\frac{1}{2}f^s_tD_kM_{sj}^{ti}\left(\eta^{jq}\delta^p_i+\delta_{i}^q\eta^{jp}\right)\theta^k+\frac{1}{2}G_{kj}^i\left(\eta^{jq}\delta^p_i+\delta_{i}^q\eta^{jp}\right)\theta^k\\
&=-\frac{1}{2}f^s_t\left(\eta^{jq}D_kM_{sj}^{tp}+D_kM_{sj}^{tq}\eta^{jp}\right)\theta^k+\frac{1}{2}\left(\eta^{jq}G_{kj}^p+G_{kj}^q\eta^{jp}\right)\theta^k.
\end{align*}

Thus, it is sufficient to take functions $G^i_{kj}$ fulfilling the equation
\begin{equation}
\eta^{jq}G_{kj}^p+G_{kj}^q\eta^{jp}=f^s_t\left(\eta^{jq}D_kM_{sj}^{tp}+D_kM_{sj}^{tq}\eta^{jp}\right)-2\mu_j^{ipq}D_kf_i^j.\label{Ges}
\end{equation}

This assures us that $\cL_Z\left(\omega_{\pf}\right)^{pq}=0$.

In order to look for a solution to \eqref{Ges}, consider the decomposition of the set of $\displaystyle {0}\choose{3}$-tensors of a vector space $V$ (introduced in  \cite[Lemma 4.3]{LewisDecomp}), i.e. 
$$
T^0_3(V)=\Lambda^3V\oplus S^3V\oplus (S_{12}V\cap \ker(\text{Sym}))\oplus (S_{23}V\cap \ker(\text{Sym})) \ ,
$$
where $A\in S_{12}V$ if, and only if, $A(u,v,w)=A(v,u,w)$, and $B\in S_{23}V$ if, and only if, $B(u,v,w)=B(u,w,v)$, for every $u,v,w\in V$ (here $\text{Sym}$ denotes the symmetrization projector). Such decomposition is given by $A=\Omega_A+S_A+R_A+T_A$, where $\Omega_A=\text{Alt}(A)\in \Lambda^3V$, $S_A=\text{Sym}(A)\in S^3V$ and
\[
R_A(u,v,w)=\frac{1}{3}\left(A(u,v,w)+A(v,u,w)-A(v,w,u)-A(u,w,v)\right)\in S_{12}V\cap \ker(\text{Sym}) \ ,
\]
\[
T_A(u,v,w)=\frac{1}{3}\left(A(u,v,w)+A(u,w,v)-A(v,u,w)-A(w,u,v)\right)\in S_{23}V\cap \ker(\text{Sym}) \ ,
\]
or using a basis for $V$
\[
(R_A)_{ijk}=\frac{1}{3}\left(A_{ijk}+A_{jik}-A_{jki}-A_{ikj}\right)\ ,
\]
\[
(T_A)_{ijk}=\frac{1}{3}\left(A_{ijk}+A_{ikj}-A_{jik}-A_{kij}\right)\ ,
\]

Now, we want to solve the equation
\begin{equation}
A_{ijk}=B_{ijk}\label{AesB}\ ,
\end{equation}
with $A_{ijk}=G^q_{ij}\eta_{kq}+G^q_{ik}\eta_{jq}$ and 
\[
B_{ijk}=f^s_t\left(\eta_{kp}D_iM_{sj}^{tp}+D_iM_{sk}^{tq}\eta_{qj}\right)-2\eta_{kp}\mu_r^{lpq}\eta_{qj}D_if_l^r, 
\]
for symmetric tensors $G^k_{ij}=G^k_{ji}$. We will use the above mentioned decomposition.

It is readily seen that $R_A=R_B=\Omega_A=\Omega_B=0$. Furthermore
\begin{align*}
(T_A)_{ijk}&=\frac{1}{3}\left(G^q_{ij}\eta_{kq}+G^q_{ik}\eta_{jq}+G^q_{ik}\eta_{jq}+G^q_{ij}\eta_{kq}-G^q_{ji}\eta_{kq}-G^q_{jk}\eta_{iq}-G^q_{ki}\eta_{jq}-G^q_{kj}\eta_{iq}\right)\\
&=\frac{1}{3}\left(G^q_{ik}\eta_{jq}+G^q_{ij}\eta_{kq}-G^q_{jk}\eta_{iq}-G^q_{kj}\eta_{iq}\right)
=\frac{1}{3}\left(A_{ijk}-2G^q_{jk}\eta_{iq}\right)
\end{align*}
and
\begin{align*}
(S_A)_{ijk}&=\frac{1}{3}\left(G^q_{ij}\eta_{kq}+G^q_{ik}\eta_{jq}+G^q_{ki}\eta_{jq}+G^q_{kj}\eta_{iq}+G^q_{jk}\eta_{iq}+G^q_{ji}\eta_{kq}\right)\\
&=\frac{2}{3}\left(G^q_{ij}\eta_{kq}+G^q_{ki}\eta_{jq}+G^q_{jk}\eta_{iq}\right)
=\frac{2}{3}\left(A_{ijk}+G^q_{jk}\eta_{iq}\right) \ ,
\end{align*}
and for $B$ we get
\begin{align*}
(T_B)_{ijk}&=\frac{1}{3}\left(B_{ijk}+B_{ikj}-B_{jik}-B_{kij}\right)\
=\frac{1}{3}\left(2B_{ijk}-B_{jik}-B_{kij}\right)\\
\end{align*}
and
\begin{align*}
(S_B)_{ijk}&=\frac{1}{3}\left(B_{ijk}+B_{kij}+B_{jki}\right)\ .
\end{align*}
Thus, equating each term, we find
\[
A_{ijk}-2G^q_{jk}\eta_{iq}=2B_{ijk}-B_{jik}-B_{kij} \ ,
\]
from $T_A=T_B$, and
\[
2A_{ijk}+2G^q_{jk}\eta_{iq}=B_{ijk}+B_{kij}+B_{jki} \ ,
\]
from $S_A=S_B$. But using the initial equation \eqref{AesB}, we get
\[
G^q_{jk}=\frac{1}{2}\eta^{iq}\left(B_{jik}+B_{kij}-B_{ijk}\right)
\]
and
\[
G^q_{jk}=\frac{1}{2}\eta^{iq}\left(B_{kij}+B_{jki}-B_{ijk}\right).
\]
In conclusion, Eq. \eqref{AesB} has solutions, one of them being given by the last equation.
\end{proof}

\begin{note}\label{rem:LocalizeGriffithsVarProblem}
A consequence of the previous Lemma is that, given a trivializing open set $V$ of $LM$, there exists an infinitesimal symmetry $Z$ of $\cI_{\text{L}}$ such that $\tau_1\left(\mathop{\text{supp}}{Z}\right)\subset\tau\left(V\right)$. It allows us to obtain local conditions for the extremals of the Griffiths variational problem.
\end{note}

\subsection{Field equations for Lovelock gravity from its Griffiths variational problem}
\label{sec:equat-moti-palat}

Let us compute the differential of the Lovelock Lagrangian:
\begin{eqnarray*}
\dif\lambda_{\text{L}}&=&\left(\dif\theta_{i_{1}\dots i_rl_1\dots l_{r}}\right)\wedge\Omega^{i_1l_1}\wedge\dots\wedge \Omega^{i_rl_r}
\\ & &
+(-1)^m\sum_{a=1}^{r}\theta_{i_{1}\dots i_rl_1\dots l_{r}}\wedge\Omega^{i_1l_1}\wedge\dots
\wedge\dif\Omega^{i_al_a}\wedge\dots\wedge \Omega^{i_rl_r}.
\end{eqnarray*}
The first term was computed in Proposition \ref{proptheta}, now let us work out the second term. If $\sigma\in S_r$ is the permutation that transpose $1$ and $a$, then 
\[
\theta_{i_{\sigma(1)}\dots i_{\sigma(r)}l_{\sigma(1)}\dots l_{\sigma(r)}}=\theta_{i_{1}\dots i_rl_1\dots l_{r}},
\]
and we can reorder every summand as follows
\begin{align*}
&\sum_{a=1}^{r}\theta_{i_{1}\dots i_rl_1\dots l_{r}}\wedge\Omega^{i_1l_1}\wedge\dots\wedge\dif\Omega^{i_al_a}\wedge\dots\wedge \Omega^{i_rl_r}\\
=&\sum_{a=1}^{r}\theta_{i_{\sigma(1)}\dots i_{\sigma(r)}l_{\sigma(1)}\dots l_{\sigma(r)}}\wedge\Omega^{i_{\sigma(1)}l_{\sigma(1)}}\wedge\dots\wedge\dif\Omega^{i_{\sigma(a)}l_{\sigma(a)}}\wedge\dots\wedge \Omega^{i_{\sigma(r)}l_{\sigma(r)}}\\
=&\sum_{a=1}^{r}\theta_{i_{1}\dots i_rl_1\dots l_{r}}\wedge\Omega^{i_al_a}\wedge\Omega^{i_2l_2}\wedge\dots\wedge\dif\Omega^{i_1l_1}\wedge\dots\wedge \Omega^{i_rl_r}\\
=&\sum_{a=1}^{r}\theta_{i_{1}\dots i_rl_1\dots l_{r}}\wedge\dif\Omega^{i_1l_1}\wedge\Omega^{i_2l_2}\wedge\dots\wedge \Omega^{i_rl_r}\\
=&\ r\theta_{i_{1}\dots i_rl_1\dots l_{r}}\wedge\dif\Omega^{i_1l_1}\wedge\Omega^{i_2l_2}\wedge\dots\wedge \Omega^{i_rl_r}.
\end{align*}
Hence
\begin{align*}
\dif\lambda_{\text{L}}=&\left(\dif\theta_{i_{1}\dots i_rj_1\dots j_{r}}\right)\wedge\Omega^{i_1j_1}\wedge\dots\wedge \Omega^{i_rj_r}+\\
\phantom{=}&(-1)^mr\theta_{i_{1}\dots i_rj_1\dots j_{r}}\wedge\dif\Omega^{i_1j_1}\wedge\Omega^{i_2j_2}\wedge\dots\wedge \Omega^{i_rj_r}\\
=&\left(T^{l}\wedge\theta_{i_{1}\dots i_rj_1\dots j_{r}l}+\sum_{s=1}^{r}(-1)^{s}\omega^{l}_{i_s}\wedge\theta_{i_1\dots \widehat{i_s}\dots i_{r}j_1\dots j_{r}l}\right.+\\
\phantom{=} &\left.\sum_{s=1}^{r}(-1)^{r+s}\omega^{l}_{j_s}\wedge\theta_{i_1\dots i_{r}j_1\dots\widehat{j_s}\dots  j_{r}l}-\omega^{l}_{l}\wedge\theta_{i_1\dots i_{r}j_1\dots j_{r}}\right)\wedge\Omega^{i_1j_1}\wedge\dots\wedge \Omega^{i_rj_r}+\\
\phantom{=}&r\dif\Omega^{i_1j_1}\wedge\theta_{i_{1}\dots i_rj_1\dots j_{r}}\wedge\Omega^{i_2j_2}\wedge\dots\wedge \Omega^{i_rj_r}\\
=&\left[\left(T^{l}\wedge\theta_{i_{1}\dots i_rj_1\dots j_{r}l}+\sum_{s=1}^{r}(-1)^{s}\omega^{l}_{i_s}\wedge\theta_{i_1\dots \widehat{i_s}\dots i_{r}j_1\dots j_{r}l}\right.\right.+\\
\phantom{=}&\left.\left.\sum_{s=1}^{r}(-1)^{r+s}\omega^{l}_{j_s}\wedge\theta_{i_1\dots i_{r}j_1\dots\widehat{j_s}\dots  j_{r}l}-\omega^{l}_{l}\wedge\theta_{i_1\dots i_{r}j_1\dots j_{r}}\right)\wedge\Omega^{i_1j_1}+\right.\\
&\left.r\left(\Omega^{i_1}_{q}\wedge\omega^{qj_1}-\omega^{i_1}_{q}\wedge\Omega^{qj_1}\right)\wedge\theta_{i_{1}\dots i_rj_1\dots j_{r}}\vphantom{\sum_{s=1}^{2r}}\right]\wedge\Omega^{i_2j_2}\wedge\dots\wedge \Omega^{i_rj_r}.
\end{align*}
Let us analyze the two sums in the brackets. As above, suppose that $\sigma$ is the permutation that transpose $1$ and $s$; then
\begin{align*}
&\sum_{s=1}^{r}(-1)^{s}\omega^{l}_{i_s}\wedge\theta_{i_1\dots \widehat{i_s}\dots i_{r}j_1\dots j_{r}l}\wedge\Omega^{i_1j_1}\wedge\dots\wedge \Omega^{i_rj_r}\\
=&\sum_{s=1}^{r}(-1)^{s}\omega^{l}_{i_{\sigma(s)}}\wedge\theta_{i_{\sigma(1)}\dots \widehat{i_{\sigma(s)}}\dots i_{\sigma(r)}j_{\sigma(1)}\dots j_{\sigma(r)}l}\wedge\Omega^{i_{\sigma(1)}j_{\sigma(1)}}\wedge\dots\wedge \Omega^{i_{\sigma(r)}j_{\sigma(r)}}\\
=&\sum_{s=1}^{r}(-1)^{s+1}\omega^{l}_{i_{1}}\wedge\theta_{i_{s}i_2\dots \widehat{i_{s}}\dots i_{r}j_{1}\dots j_{r}l}\wedge\Omega^{i_{1}j_{1}}\wedge\dots\wedge \Omega^{i_{r}j_{r}}\\
=&-\sum_{s=1}^{r}\omega^{l}_{i_{1}}\wedge\theta_{i_2\dots i_{r}j_{1}\dots j_{r}l}\wedge\Omega^{i_{1}j_{1}}\wedge\dots\wedge \Omega^{i_{r}j_{r}}\\
=&-r\omega^{l}_{i_{1}}\wedge\theta_{i_2\dots i_{r}j_{1}\dots j_{r}l}\wedge\Omega^{i_{1}j_{1}}\wedge\dots\wedge \Omega^{i_{r}j_{r}}.
\end{align*}
Furthermore, using a similar argument,
\begin{align*}
&\sum_{s=1}^{r}(-1)^{r+s}\omega^{l}_{j_s}\wedge\theta_{i_1\dots i_{r}j_1\dots\widehat{j_s}\dots  j_{r}l}\wedge \Omega^{i_1j_1}\wedge\dots\wedge \Omega^{i_rj_r}\\
=&\sum_{s=1}^{r}(-1)^{r+s}\omega^{l}_{j_{\sigma(s)}}\wedge\theta_{i_{\sigma(1)}\dots i_{\sigma(r)}j_{\sigma(1)}\dots \widehat{j_{\sigma(s)}}\dots j_{\sigma(r)}l}\wedge\Omega^{i_{\sigma(1)}j_{\sigma(1)}}\wedge\dots\wedge \Omega^{i_{\sigma(r)}j_{\sigma(r)}}\\
=&\sum_{s=1}^{r}(-1)^{r+s+1}\omega^{l}_{j_{1}}\wedge\theta_{i_1\dots i_{r}j_{s}j_2\dots \widehat{j_{s}}\dots j_{r}l}\wedge\Omega^{i_{1}j_{1}}\wedge\dots\wedge \Omega^{i_{r}j_{r}}\\
=&\sum_{s=1}^{r}(-1)^{r+1}\omega^{l}_{j_{1}}\wedge\theta_{i_1\dots i_{r}j_2\dots j_{r}l}\wedge\Omega^{i_{1}j_{1}}\wedge\dots\wedge \Omega^{i_{r}j_{r}}\\
=&\ r\omega^{l}_{j_{1}}\wedge\theta_{i_1\dots i_{r}lj_2\dots j_{r}}\wedge\Omega^{i_{1}j_{1}}\wedge\dots\wedge \Omega^{i_{r}j_{r}}\\
=&\ r\eta^{pq}\omega^{l}_{p}\wedge\theta_{i_1\dots i_{r}lj_2\dots j_{r}}\wedge\Omega^{i_{1}}_{q}\wedge\Omega^{i_{2}j_{2}}\wedge\dots\wedge \Omega^{i_{r}j_{r}}.
\end{align*}
Now let us simplify the terms $r\left(\Omega^{i_1}_{q}\wedge\omega^{qj_1}-\omega^{i_1}_{q}\wedge\Omega^{qj_1}\right)\wedge\theta_{i_{1}\dots i_rj_1\dots j_{r}}\wedge\Omega^{i_{2}j_{2}}\wedge\dots\wedge \Omega^{i_{r}j_{r}}$. First, renaming the dummy indices,
\begin{align*}
&-r\omega^{l}_{i_1}\wedge\Omega^{i_1j_1}\wedge\theta_{li_2\dots i_rj_1\dots j_{r}}\wedge\Omega^{i_{2}j_{2}}\wedge\dots\wedge \Omega^{i_{r}j_{r}}
=r\omega^{l}_{i_1}\wedge\theta_{i_{2}\dots i_rj_1\dots j_{r}l}\wedge\Omega^{i_1j_1}\wedge\Omega^{i_{2}j_{2}}\wedge\dots\wedge \Omega^{i_{r}j_{r}},
\end{align*}
which cancels out the first sum. Second
\begin{align*}
r\Omega^{i_1}_{q}\wedge\omega^{qj_1}\wedge\theta_{i_1\dots i_rj_1\dots j_{r}}\wedge\Omega^{i_{2}j_{2}}\wedge\dots\wedge \Omega^{i_{r}j_{r}}
=&r\eta^{j_1p}\omega^{q}_p\wedge\theta_{i_1\dots i_rj_1\dots j_{r}}\wedge\Omega^{i_1}_{q}\wedge\Omega^{i_{2}j_{2}}\wedge\dots\wedge \Omega^{i_{r}j_{r}}\\
=&r\eta^{lp}\omega^{q}_p\wedge\theta_{i_1\dots i_rl\dots j_{r}}\wedge\Omega^{i_1}_{q}\wedge\Omega^{i_{2}j_{2}}\wedge\dots\wedge \Omega^{i_{r}j_{r}},
\end{align*}
and consequently
\begin{align*}
\dif\lambda_{\text{L}}&=\left(\eta^{j_1q}T^{l}\wedge\theta_{i_{1}\dots i_rj_1\dots j_{r}l}-\eta^{j_1q}\omega^{l}_{l}\wedge\theta_{i_1\dots i_{r}j_1\dots j_{r}}\right.+\\
&\left.
r\left(\eta^{j_1p}\omega^{q}_p+\eta^{qp}\omega^{j_1}_p\right)\wedge\theta_{i_1\dots i_rj_1\dots j_{r}}\right)\wedge\Omega^{i_1}_q\wedge\Omega^{i_2j_2}\wedge\dots\wedge \Omega^{i_rj_r}.
\end{align*}
These computations amounts to the Lagrangian form on $J^1\tau$, so we have to take its pullback to $T_0$, i.e.
\begin{align*}
\iota_0^*\dif\lambda_{\text{L}}=&\left[r\left(\eta^{j_1p}\omega^{q}_p+\eta^{qp}\omega^{j_1}_p\right)\wedge\theta_{i_1\dots i_rj_1\dots j_{r}}-\eta^{j_1q}\omega^{l}_{l}\wedge\theta_{i_1\dots i_{r}j_1\dots j_{r}}\right]\wedge\Omega^{i_1}_q\wedge\Omega^{i_2j_2}\wedge\dots\wedge \Omega^{i_rj_r}\\
=&2\left[r\eta^{j_1p}\left(\omega_{\mathfrak p}\right)^{q}_p\wedge\theta_{i_1\dots i_rj_1\dots j_{r}}-\frac{1}{2}\eta^{j_1q}\left(\omega_{\mathfrak p}\right)^{l}_{l}\wedge\theta_{i_1\dots i_{r}j_1\dots j_{r}}\right]\wedge\Omega^{i_1}_q\wedge\Omega^{i_2j_2}\wedge\dots\wedge \Omega^{i_rj_r} \ .
\end{align*}
Nevertheless, we will omit the pullback $\iota_0$ to simplify notation.

Now we are ready to find the field equations associated with the Griffiths problem $\left(J^1\tau,\lambda_{\text{L}},\cI_{\text{L}}\right)$. First we state a lemma we will use later on.
\begin{lem}\label{lem:omegacomT}
If $\Omega$ takes values in $\mathfrak k$, then on $T_0$
\[
\Omega^q_{i_1}\wedge\theta_{qi_2\dots i_r J}\wedge\Omega^{I J}=-\Omega^q_{j_{1}}\wedge\theta_{Iq j_2\dots j_{r} }\wedge\Omega^{I J}.
\]
\end{lem}
\begin{proof}
Using the structure equation, we have
\[
\dif T^q+\omega^q_l\wedge T^l=\Omega^q_l\wedge\theta^l \ ,
\]
and since the torsion form annihilates on $T_0$, we have
\[
\Omega^q_l\wedge\theta^l=0.
\]
Then multiplying both sides by $\theta_{qi_1\dots i_rj_1\dots j_r}\wedge\Omega^{IJ}$,
\[
\Omega^q_l\wedge\theta^l\wedge\theta_{qi_1\dots i_rj_1\dots j_r}\wedge\Omega^{IJ}=0,
\]
hence, using Proposition \ref{proptheta},
\begin{align*}
&\Omega^q_l\wedge\theta^l\wedge\theta_{qIJ}\wedge\Omega^{IJ} \\
&=(-1)^{2r+1}\Omega^q_l\wedge\left[-\delta^l_q\theta_{IJ}+{\sum_{a=1}^r}\sum_{a=1}^r(-1)^{a+1}\left(\delta^l_{i_a}\theta_{qi_1\dots\hat{i}_a\dots i_r J}+(-1)^r\delta^l_{j_a}\theta_{qIj_1\dots\hat{j}_a\dots j_r }\right)\right]\wedge\Omega^{IJ}\\
&=\Omega^q_l\wedge\left[\delta^l_q\theta_{IJ} +{\sum_{a=1}^r}\sum_{a=1}^r(-1)^{a}\left(\delta^l_{i_a}\theta_{qi_1\dots\hat{i}_a\dots i_r J}+(-1)^r\delta^l_{j_a}\theta_{qIj_1\dots\hat{j}_a\dots j_r }\right)\right]\wedge\Omega^{IJ}\\
&=\left[\Omega^l_l\wedge\theta_{IJ} +{\sum_{a=1}^r}\sum_{a=1}^r(-1)^{a}\left(\Omega^q_{i_a}\wedge\theta_{qi_1\dots\hat{i}_a\dots i_r J}+(-1)^r\Omega^q_{j_a}\wedge\theta_{qIj_1\dots\hat{j}_a\dots j_r }\right)\right]\wedge\Omega^{IJ}.
\end{align*}
Let us study both terms in the sum. If $\sigma\in S_r$ is the permutation transposing $1$ and $a$,
\begin{align*}
&\sum_{a=1}^r(-1)^{a}\Omega^q_{i_a}\wedge\theta_{qi_1\dots\hat{i}_a\dots i_r J}\wedge\Omega^{IJ}=\sum_{a=1}^r(-1)^{a}\Omega^q_{i_{\sigma(a)}}\wedge\theta_{qi_{\sigma(1)}\dots\hat{i}_{\sigma(a)}\dots i_{\sigma(r)} J_\sigma}\wedge\Omega^{I_\sigma J_\sigma}\\
=&\sum_{a=1}^r(-1)^{a+1}\Omega^q_{i_1}\wedge\theta_{qi_ai_2\dots\hat{i}_a\dots i_r J}\wedge\Omega^{I J}=\sum_{a=1}^r(-1)^{a+1}(-1)^{a-2}\Omega^q_{i_1}\wedge\theta_{qi_2\dots i_r J}\wedge\Omega^{I J}\\
=&\sum_{a=1}^r-\Omega^q_{i_1}\wedge\theta_{qi_2\dots i_r J}\wedge\Omega^{I J}=-r\Omega^q_{i_1}\wedge\theta_{qi_2\dots i_r J}\wedge\Omega^{I J},
\end{align*}
where we use the skew-symmetry of the Sparling form and the fact that the curvature is a two-form. On the other hand
\begin{align*}
&\sum_{a=1}^r(-1)^{a+r}\Omega^q_{j_a}\wedge\theta_{qIj_1\dots\hat{j}_a\dots j_r }\wedge\Omega^{IJ}=\sum_{a=1}^r(-1)^{a+r}\Omega^q_{j_{\sigma(a)}}\wedge\theta_{qI_\sigma j_{\sigma(1)}\dots\hat{j}_{\sigma(a)}\dots j_{\sigma(r)} }\wedge\Omega^{I_\sigma J_\sigma}\\
=&\sum_{a=1}^r(-1)^{a+r+1}\Omega^q_{j_{1}}\wedge\theta_{qI j_{a}j_2\dots\hat{j}_{1}\dots j_{r} }\wedge\Omega^{I J}=\sum_{a=1}^r(-1)^{a+r+1}(-1)^{a-2}\Omega^q_{j_{1}}\wedge\theta_{qI j_2\dots j_{r} }\wedge\Omega^{I J}\\
=&(-1)^{r+1}r\Omega^q_{j_{1}}\wedge\theta_{qI j_2\dots j_{r} }\wedge\Omega^{I J}\ .
\end{align*}
Thus, since $\Omega^l_l=0$, we get
\[
\Omega^q_{i_1}\wedge\theta_{qi_2\dots i_r J}\wedge\Omega^{I J}=(-1)^{r+1}\Omega^q_{j_{1}}\wedge\theta_{qI j_2\dots j_{r} }\wedge\Omega^{I J}\ ,
\]
or equivalently
\[
\Omega^q_{i_1}\wedge\theta_{qi_2\dots i_r J}\wedge\Omega^{I J}=-\Omega^q_{j_{1}}\wedge\theta_{Iq j_2\dots j_{r} }\wedge\Omega^{I J} \ .\qedhere
\]
\end{proof}

As we did in the proof of Lemma \ref{lem:InfSymmetryLift}, consider a basis $\left\{U_j^i\right\}$ of $G$-invariant local vector fields generating $\mathfrak X^{V\tau}(V)$ on some trivializing open set $V\subset LM$ and let $\left\{M_{ij}^{kl}\right\}$ be smooth functions on $V$ such that
\[
  U_j^i=M_{jk}^{il}\left(E_l^k\right)_{LM}.
\]

Let us compute the contraction $j^1U_r^t\lrcorner \dif\lambda_{\text{L}}$ by using \eqref{jotaU}. First
\begin{align*}
M_{rk}^{ti}\left(E_i^k\right)_{J^1\tau}\lrcorner \dif\lambda_{\text{L}}=&j^1U_r^t\lrcorner2\left[r\eta^{j_1p}\left(\omega_{\mathfrak p}\right)^{q}_p\wedge\theta_{IJ}-\frac{1}{2}\eta^{j_1q}\left(\omega_{\mathfrak p}\right)^{l}_{l}\wedge\theta_{IJ}\right]\wedge\Omega^{i_1}_q\wedge
\Omega^{i_2j_2}\wedge\dots\wedge \Omega^{i_rj_r}\\
=&2M^{ti}_{rk}\left[r\eta^{j_1p}\frac{1}{2}\left(\delta^q_i\delta^k_p+\eta_{pa}\delta^{k}_b\delta^a_i\eta^{bq}\right)\wedge\theta_{IJ}-\frac{1}{2}\eta^{j_1q}\delta^{k}_{i}\theta_{IJ}\right]\wedge\Omega^{i_1}_q\wedge
\Omega^{i_2j_2}\wedge\dots\wedge \Omega^{i_rj_r}\\
=&2M^{ti}_{rk}\left[r\frac{1}{2}\left(\eta^{j_1k}\delta^q_i+\delta^{j_1}_i\eta^{kq}\right)\wedge\theta_{IJ}-\frac{1}{2}\eta^{j_1q}\delta^{k}_{i}\theta_{IJ}\right]\wedge\Omega^{i_1}_q\wedge
\Omega^{i_2j_2}\wedge\dots\wedge \Omega^{i_rj_r}\\
=&2M^{ti}_{rk}\left[r\frac{1}{2}\left(\Omega^{i_1}_i\eta^{j_1k}+\delta^{j_1}_i\eta^{kq}\Omega^{i_1}_q\right)\wedge\theta_{IJ}-\frac{1}{2}\eta^{j_1q}\delta^{k}_{i}\theta_{IJ}\wedge\Omega^{i_1}_q\right]
\wedge\Omega^{i_2j_2}\wedge\dots\wedge \Omega^{i_rj_r}
\end{align*}
Let us study the terms in brackets. First
\begin{align*}
&\left(\Omega^{i_1}_i\eta^{j_1k}+\delta^{j_1}_i\eta^{kq}\Omega^{i_1}_q\right)\wedge\theta_{IJ}\wedge\Omega^{I'J'}
=\eta^{j_1k}\Omega^{i_1}_i\wedge\theta_{IJ}\wedge\Omega^{I'J'}+\delta^{j_1}_i\eta^{kq}\Omega^{i_1}_q\wedge\theta_{IJ}\wedge\Omega^{I'J'},
\end{align*}
and using Lemma \ref{lem:omegacomT} and renaming indices,
\begin{align*}
&\eta^{j_1k}\Omega^{i_1}_i\wedge\theta_{IJ}\wedge\Omega^{I'J'}+\delta^{j_1}_i\eta^{kq}\Omega^{i_1}_q\wedge\theta_{IJ}\wedge\Omega^{I'J'}\\
=&\ \eta^{j_1k}\Omega^{q}_i\wedge\theta_{qi_2\dots i_rj_1\dots j_r}\wedge\Omega^{I'J'}+\eta^{kq}\Omega^{i_1}_q\wedge\theta_{i_1\dots i_rij_2\dots j_r}\wedge\Omega^{I'J'}\\
=&\ -\eta^{j_1k}\Omega^{q}_{j_1}\wedge\theta_{ii_2\dots i_rqj_2\dots j_r}\wedge\Omega^{I'J'}+\eta^{kq}\Omega^{i_1}_q\wedge\theta_{i_1\dots i_rij_2\dots j_r}\wedge\Omega^{I'J'}\\
=&\ \eta^{kq}\Omega^{i_1}_{q}\wedge\theta_{i_1\dots i_rij_2\dots j_r}\wedge\Omega^{I'J'}+\eta^{kq}\Omega^{i_1}_q\wedge\theta_{i_1\dots i_rij_2\dots j_r}\wedge\Omega^{I'J'}\\
=&\ 2\eta^{kq}\Omega^{i_1}_{q}\wedge\theta_{i_1\dots i_rij_2\dots j_r}\wedge\Omega^{I'J'}.
\end{align*}
Thus
\begin{align*}
M_{rk}^{ti}\left(E_i^k\right)_{J^1\tau}\lrcorner \dif\lambda_{\text{L}}=&2M^{ti}_{rk}\left[r\eta^{kq}\Omega^{i_1}_{q}\wedge\theta_{Iij_2\dots j_r}-\frac{1}{2}\eta^{j_1q}\delta^{k}_{i}\theta_{IJ}\wedge\Omega^{i_1}_q\right]\wedge\Omega^{I'J'}\\
=&2rM^{ti}_{rk}\left[\eta^{kq}\theta_{Iij_2\dots j_r}-\frac{1}{2r}\eta^{j_1q}\delta^{k}_{i}\theta_{IJ}\right]\wedge\Omega^{i_1}_{q}\wedge\Omega^{I'J'}\ .
\end{align*}
Now let us compute the other contraction:
\begin{align*}
&D_jM_{rk}^{ti}\left(\theta^j,\left(E_i^k\right)_{LM}\right)^V\lrcorner \dif\lambda_{\text{L}}\\
=&D_jM_{rk}^{ti}\left(\theta^j,\left(E_i^k\right)_{LM}\right)^V\lrcorner 2\left[r\eta^{j_1p}\left(\omega_{\mathfrak p}\right)^{q}_p\wedge\theta_{IJ}-\frac{1}{2}\eta^{j_1q}\left(\omega_{\mathfrak p}\right)^{l}_{l}\wedge\theta_{IJ}\right]\wedge\Omega^{i_1}_q\wedge\Omega^{I'J'}\\
=&\left(-1\right)^{m+1}D_jM_{rk}^{ti}2\left[r\eta^{j_1p}\left(\omega_{\mathfrak p}\right)^{q}_p\wedge\theta_{IJ}-\frac{1}{2}\eta^{j_1q}\left(\omega_{\mathfrak p}\right)^{l}_{l}\wedge\theta_{IJ}\right]\wedge \widetilde{\Omega}^ {jk,i_1I'J'}_{i,q},
\end{align*}
where $\widetilde{\Omega}^{jk,i_1I'J'}_{i,q}=\left(\theta^j,\left(E_i^k\right)_{LM}\right)^V\lrcorner\left(\Omega^{i_1}_q\wedge\Omega^{I'J'}\right)$. 
So, gathering both terms together:
\begin{align*}
j^1U_r^t\lrcorner \dif\lambda_{\text{L}}=&2rM^{ti}_{rk}\left[\eta^{kq}\theta_{Iij_2\dots j_r}-\frac{1}{2r}\eta^{j_1q}\delta^{k}_{i}\theta_{IJ}\right]\wedge\Omega^{i_1}_{q}\wedge\Omega^{I'J'}\\
&+\left(-1\right)^{m+1}D_jM_{rk}^{ti}2\left[r\eta^{j_1p}\left(\omega_{\mathfrak p}\right)^{q}_p\wedge\theta_{IJ}-\frac{1}{2}\eta^{j_1q}\left(\omega_{\mathfrak p}\right)^{l}_{l}\wedge\theta_{IJ}\right]\wedge \widetilde{\Omega}^ {jk,i_1I'J'}_{i,q} \ .
\end{align*}
Thus if $\Sigma:M\rightarrow T_0$ is a section such that $\Sigma^*\omega_\pf=0$, we can conclude that
\begin{equation}\label{eq:VariationHorizontal}
 \Sigma^*\left(j^1U_r^t\lrcorner \dif\lambda_{\text{L}}\right)=2r\left(M^{ti}_{rk}\circ\Sigma\right)\Sigma^*\left[\left(\eta^{kq}\theta_{Iij_2\dots j_r}-\frac{1}{2r}\eta^{j_1q}\delta^{k}_{i}\theta_{IJ}\right)\wedge\Omega^{i_1}_{q}\wedge\Omega^{I'J'}\right].
\end{equation}

From Lemma \ref{lem:InfSymmetryLift} and Equation \eqref{eq:VariationHorizontal}, we obtain the following result:

\begin{thm}\label{thm:EinsteinEqsNeccesary}
  Let $\Sigma:M\rightarrow T_0$ be an extremal for the variational problem associated with the Griffiths triple $\left(T_0,\lambda_{\text{L}},\cI_{\text{L}}\right)$. Then
  \[
    \Sigma^*\left[\left(\eta^{kq}\theta_{Iij_2\dots j_r}-\frac{1}{2r}\eta^{j_1q}\delta^{k}_{i}\theta_{IJ}\right)\wedge\Omega^{i_1}_{q}\wedge\Omega^{I'J'}\right]=0.
  \]
\end{thm}
\begin{proof}
  We consider infinitesimal symmetries $Z\in\mathfrak{X}^{V\tau_1}\left(T_0\right)$ of $\cI_{\text{L}}$ as in Lemma \ref{lem:InfSymmetryLift}, namely
  \[
    Z=f^l_jj^1U^j_l+F^i_{kl}\left(\theta^k,\left(E^l_i\right)_{LM}\right)^V\ ,
  \]
  where $(f^l_j)$ is a family of arbitrary functions on (an open set of) $M$ such that
  \[
    \mathop{\text{supp}}{f_j^i}\subset\tau\left(V\right).
  \]
  Then by performing the variation induced by $Z$, we have the formula
  \[
    \int_{\tau\left(V\right)}f_r^s\Sigma^*\left[M^{ri}_{sk}\left(\eta^{kq}\theta_{Iij_2\dots j_r}-\frac{1}{2r}\eta^{j_1q}\delta^{k}_{i}\theta_{IJ}\right)\wedge\Omega^{i_1}_{q}\wedge\Omega^{I'J'}\right]=0,
  \]
  and the result follows from the fact that the functions $f^i_k$ are arbitrary and the matrix $\left(M_{jk}^{il}\right)$ is invertible.
\end{proof}

\begin{note}
  It is useful to compare this with the Einstein case. In \cite{doi:10.1142/S0219887818500445} it is seen that the Einstein equations in vacuum are
  \[
    \theta_{il}\wedge\Omega^l_k+\theta_{kl}\wedge\Omega^l_i-\eta_{ik}\left(\eta^{pq}\theta_{ql}\wedge\Omega^l_p\right)=0,
  \]
  together with the constraints $T=0=\omega_\pf$; the previous theorem, on the other hand, gives us the set of equations
  \[
    \left(\eta^{kp}\theta_{il}-\frac{1}{2}\delta^k_i\eta^{qp}\theta_{ql}\right)\wedge\Omega_p^l=0
  \]
  under the same constraints. Nevertheless, it can be proved (see Corollary $21$ in \cite{doi:10.1142/S0219887818500445}) that under the constraints $T=0=\omega_\pf$, it is true that
  \[
    \omega_{ik}\wedge\Omega^k_p-\omega_{pk}\wedge\Omega^k_i=0
  \]
  as consequence of a Bianchi identity. Therefore, these sets of equations are equivalent.
\end{note}

Theorem \ref{thm:EinsteinEqsNeccesary} gives us a set of necessary conditions for a section $\Sigma:M\rightarrow T_0$ to be extremal of the Griffiths variational problem associated with the triple $\left(T_0,\lambda_{\text{L}},\cI_{\text{L}}\right)$. Our next task is to set the sufficiency of these conditions.

\begin{prop}
  Let $\Sigma:M\rightarrow T_0$ be a section such that $\Sigma^*\omega_\pf=0$ and
  \[
    \Sigma^*\left[\left(\eta^{kq}\theta_{Iij_2\dots j_r}-\frac{1}{2r}\eta^{j_1q}\delta^{k}_{i}\theta_{IJ}\right)\wedge\Omega^{i_1}_{q}\wedge\Omega^{I'J'}\right]=0.
  \]
  Then
  \[
    \Sigma^*\left(Z\lrcorner \dif\lambda_{\text{L}}\right)=0\ ,
  \]
  for every $Z\in\mathfrak{X}^{V\tau_1}\left(T_0\right)$.
\end{prop}
\begin{proof}
  It is a consequence of the fact that every $Z\in\mathfrak{X}^{V\tau_1}\left(T_0\right)$ can be written in terms of the vector fields
  \[
    \left\{j^1U^j_i,\left(\theta^k,(E^i_j)_{LM}\right)^{V}\right\}.\qedhere
  \]
\end{proof}

Thus, in particular, for every $\tau_1$-vertical infinitesimal symmetry $Z$ of the exterior differential system $\cI_{\text{L}}$ and any section $\Sigma:M\rightarrow J^1\tau$ fulfilling the hypotheses in the previous proposition, we have that
\[
  \Sigma^*\left(Z\lrcorner \dif\lambda_{\text{L}}\right)=0\ ;
\]
therefore, $\Sigma$ is an extremal for the Griffiths variational problem $\left(T_0,\lambda_{\text{L}},\cI_{\text{L}}\right)$, as required.

\section{Unified formalism}
\label{unified}

\subsection{Tautological form on a bundle of forms}\label{sec:tauto}

The next definitions are quoted from \cite{doi:10.1142/S0219887818500445}. 
Let $\pi:P\rightarrow N$ be a principal fiber bundle with structure group $H$ and assume that $q$ and $p$ are surjective submersions fitting the diagram
\[
\begin{tikzcd}
  P\ar[dr,"\psi"']\ar[rr,"\pi"]&&N\ar[dl,"\chi"]\\
&M&
\end{tikzcd}
\]

Let $V$ be a finite dimensional real vector space $V$ and define the bundle $\bar\tau^k_{n,q}:\bigwedge^k_{n,q}T^*P\otimes V\rightarrow P$ of $V$-valued $k$-forms that annihilates when contracted with $n$ $q$-vertical vectors. This bundle has a canonical $V$-valued $k$-form $\Theta^k_{n,q}$ defined through the formula
\[
\left.\Theta^k_{n,q}\right|_{\alpha}(Z_1,\dots,Z_k):=\alpha(T_\alpha \bar\tau^k_{n,q}(Z_1),\dots,T_\alpha \bar\tau^k_{n,q}(Z_k)).
\]

Given a $H$-representation $(V,\rho)$, it is readily seen that $\bigwedge^k_{n,q}T^*P\otimes V$ is a $H$-space with action given by
\[
\Phi^k_g(\alpha)(X_1,\dots,X_k):=\rho(\alpha)\cdot \left(\alpha(T_{u\cdot h}R_{h^{-1}}X_1,\dots,T_{u\cdot h}R_{h^{-1}}X_k)\right),
\]
where $R$ is the right action in $P$ and $h\in H$. It can be proved that the tautological form $\Theta^k_{n,q}$ is then a $H$-equivariant map.

We point out two instances that will be used in the next section. If $H=G=Gl(m)$, $P=J^1\tau$, $N=C(LM)$, $\psi=\tau_1$, $\chi=p$ and $\pi=q$ (that is, the left triangle in the diagram of Section \ref{sec:basic}), we have
\begin{enumerate}
\item Set $k=m-2$ and $n=r+1$, and consider $V_1=(\mathbb R^m)^*$ and $\rho_1$ the natural representation of $G$ on this vector space. Then, we denote the space $E_1:=\bigwedge^{m-2}_{2,\tau_1}J^1\tau\otimes (\mathbb R^m)^*$ and the projection
\[
p_1:E_1\rightarrow J^1\tau\ .
\]
\item Set $k=m-1$ and $n=r$, and consider $V_2=(\mathbb R^m)^*\odot (\mathbb R^m)^*$ and $\rho_2$  the natural representation of $G$ on this vector space.
(The symbol $\odot$ denotes the symmetrized tensor product). Then, we denote the space $E_2:=\bigwedge^{m-1}_{1,\tau_1}J^1\tau\otimes \left((\mathbb R^m)^*\odot (\mathbb R^m)^*\right)$ and the projection
\[
p_2:E_2\rightarrow J^1\tau\ .
\]
To simplify notation we denote by $\Theta_1$ and $\Theta_2$ the corresponding tautological forms on these bundles and, when using the component forms with respect to the canonical bases, we simply write $\Theta_1=\Theta_{l}\mathbf e^l$ and $\Theta_2=\Theta_{ij}\mathbf e^i\odot \mathbf e^j$, indicating that a single index corresponds to $\Theta_1$ and two indices to $\Theta_2$.
\end{enumerate}

\subsection{The multimomentum bundle}

The unified formalism for a Griffiths variational problem is built from the idea of a Lepage equivalent problem. Briefly, the construction goes as follows (this idea is inspired in the work of \cite{GotayCartan} and was proposed in \cite{2013arXiv1309.4080C}): assume that the differential ideal is locally generated by a subbundle $I\subset\Lambda^\bullet T^*J^1\tau$
(this means that there is an open cover $\{U_\lambda\}$ such that every $\alpha\in\mathcal I$ can be generated by sections of $\left.I\right|_{U_\lambda}$ when pulled back to $U_\lambda$). Consider an integer $k$ such that $\lambda_{\text{L}}(u)\in\Lambda^m_k(T^*_uJ^1\tau)$ and $I^m_u:=I\cap \Lambda^m_k(T^*_uJ^1\tau)\subset \Lambda^m_k(T^*_uJ^1\tau)$, where $\Lambda^m_k(T^*J^1\tau)$ is the bundle of $m$-forms that annihilate when contracted with $k$ $\tau_1$-vertical vectors. Then define the multimomentum bundle $W_\lambda$ by the equation

\[
\left.W_\lambda\right|_u=\lambda_{\text{L}}\left(u\right)+I^m_u ,
\]
which is an affine subbundle of $\Lambda^m_k(T^*J^1\tau)$. In the case of Lovelock gravity for a polynomial Lagrangian of degree $r$ in the curvature, it suffices to take $k=r+1$.
(Notice that the form $\theta_{i_1j_1\dots i_rj_r}$ is $\tau_1$-horizontal whereas the form 
$\Omega^{i_1j_1}\wedge\dots\wedge\Omega^{i_rj_r}$ is only $q$-horizontal, which implies that more than $r$ $\tau_1$-vertical vectors are needed to annihilate $\lambda_{\text{L}}$). In this case we can write any $\rho\in W_{\lambda}|_u$ as
\[
\rho=\lambda_{\text{L}}|_{j^1_xs}+\gamma_l\wedge T^l|_{j^1_xs}+\beta_{ij}\wedge \omega^{ij}_{\mathfrak p}|_{j^1_xs},
\]
where $\beta_{ij}\in\Lambda^{m-1}_r\left(T^*_{j^1_xs}J^1\tau\right)$ 
 is symmetric in $ij$ and $\gamma_l\in\Lambda^{m-2}_{r+1}\left(T^*_{j^1_xs}J^1\tau\right)$. 

Observe that an element $\rho$ in $W_\lambda$ is completely determined by an element $j^1_xs\in J^1\tau$ and the forms $\gamma_l$ and $\beta_{ij}$ projecting onto $j^1_xs$. Hence we have the following identification:

\begin{lem}\label{lem:MultiDecomposition}
The map
\[
\Gamma:\rho\mapsto (\gamma_l\mathbf e^l,\beta_{ij}\mathbf e^i\odot\mathbf e^j)\simeq (j^1_xs,\gamma_l\mathbf e^l,\beta_{ij}\mathbf e^i\odot\mathbf e^j) \ ,
\]
where $j^1_xs:=p_1(\gamma_l\mathbf e^l)=p_2(\beta_{ij}\mathbf e^i\odot\mathbf e^j)$, induces an isomorphism of the bundles $\tau_\lambda: W_\lambda\rightarrow J^1\tau$ and $pr_0: \widehat W\rightarrow J^1\tau$, 
with $\widehat W_\lambda:=E_1\times_{J^1\tau} E_2\simeq J^1\tau\times_{J^1\tau}E_1\times_{J^1\tau} E_2$, such that
\[
\begin{tikzcd}
  W_\lambda\ar[dr,"\tau_\lambda"']\ar[rr,"\Gamma"]&&\widehat W_\lambda\ar[dl,"pr_0=p_1\circ pr_1=p_2\circ pr_2"]\\
&J^1\tau&
\end{tikzcd}
\]
\end{lem}

To build the canonical form on $\widehat W_\lambda$ we need the tautological forms $\Theta_1$ and $\Theta_2$, as well as the forms $\theta, \omega,T$ and $\Omega$. We use the three projections $pr_0$, $pr_1$, and $pr_2$ to pull these forms back to $\widehat W_\lambda$, but we do not change the symbols so as to keep notation as simple as possible, e.g. if $\rho\in\widehat W_\lambda$, then $\omega|_\rho$ and $\Theta_2|_\rho$ must be understood as $pr_0^*(\omega)|_\rho=\omega|_{u=pr_0(\rho)}(T_\rho pr_0(\cdot),\dots,T_\rho pr_0(\cdot))$ and $pr_2^*(\Theta_2)|_\rho=\Theta_2|_{\beta=p_2(\rho)}(T_\rho pr_2(\cdot),\dots,T_\rho pr_2(\cdot))$, respectively.

Let us now denote by $\Theta_\lambda$ the pullback of the tautological form $\Theta\in\Omega^m(\Lambda^m(J^1\tau))$ to $\widehat W_\lambda$. Then, at any $\rho=\lambda_{\text{L}}|_{j^1_xs}+\gamma_l\wedge T^l|_{j^1_xs}+\beta_{ij}\wedge \omega^{ij}_{\mathfrak p}|_{j^1_xs},$
\[
\Theta_\lambda|_\rho=\lambda_{\text{L}}|_\rho+\Theta_l|_\rho\wedge T^l|_\rho+\Theta_{ij}|_\rho\wedge \omega^{ij}_{\mathfrak p}|_\rho \ ,
\]
or, omitting the symbol $\rho$ and recalling the expression for $\lambda_{\text{L}}$,
\[
\Theta_\lambda=\theta_{IJ}\wedge \Omega^{IJ}+\Theta_l\wedge T^l+\Theta_{ij}\wedge \omega^{ij}_{\mathfrak p}.
\]

In consequence, the differential is given by
\begin{align*}
\Omega_\lambda&=\dif\Theta_\lambda=\dif\lambda_{\text{L}}+T^l\wedge\dif\Theta_l +\dif T^l\wedge \Theta_l+\dif\Theta_{ij}\wedge \omega^{ij}_{\mathfrak p}+(-1)^{m-1}\Theta_{ij}\wedge \dif\omega^{ij}_{\mathfrak p}\\
&=\dif\lambda_{\text{L}}+T^l\wedge\dif\Theta_l +(\Omega^l_k\wedge\theta^k-\omega^l_k\wedge T^k)\wedge \Theta_l+\dif\Theta_{ij}\wedge \omega^{ij}_{\mathfrak p}\\
&\phantom{=}+(-1)^{m-1}\Theta_{ij}\wedge (-\omega^i_k\wedge\omega^{kj}+\Omega^{ij})\ ,
\end{align*}
and recalling the expression of $\dif\lambda_{\text{L}}$,
\begin{align*}
\dif\lambda_{\text{L}}&=\left[\eta^{j_1q}T^{l}\wedge\theta_{IJl}-\eta^{j_1q}\omega^{l}_{l}\wedge\theta_{IJ}+r\left(\eta^{j_1p}\omega^{q}_p+\eta^{qp}\omega^{j_1}_p\right)\wedge\theta_{IJ}\right]\wedge\Omega^{i_1}_q\wedge\Omega^{I'J'}\\
&=\left[\eta^{tq}T^{l}\wedge\theta_{stI'J'l}\wedge\Omega^{I'J'}-\eta^{tq}\omega^{l}_{l}\wedge\theta_{sj_1I'J'}\wedge\Omega^{I'J'}+\right.\\
&\phantom{=}\left.r\left(\eta^{tp}\omega^{q}_p+\eta^{qp}\omega^{t}_p\right)\wedge\theta_{stI'J'}\wedge\Omega^{I'J'}\right]\wedge\Omega^{s}_q\ ,
\end{align*}
we can rewrite
\begin{align*}
\Omega_{\lambda}&=\left[\eta^{tq}T^{l}\wedge\theta_{stI'J'l}\wedge\Omega^{I'J'}-\eta^{tq}(\omega_{\mathfrak p})^{l}_{l}\wedge\theta_{stI'J'}\wedge\Omega^{I'J'}+\theta^q\wedge\Theta_s+(-1)^{m-1}\eta^{tq}\Theta_{st}\right.\\
&\phantom{=}+2r\eta^{qp}(\omega_{\mathfrak p})^{t}_p\left.\wedge\theta_{stI'J'}\wedge\Omega^{I'J'}\right]\wedge\Omega^{s}_q+T^l\wedge(\dif\Theta_l-\omega^k_l\wedge \Theta_k)\\
&\phantom{=}+(-1)^{m}\eta^{jl}\left[(-1)^m\dif\Theta_{ij}+\Theta_{pj}\wedge(\omega_{\mathfrak k})^p_i-\eta^{qk}\eta_{pj}\Theta_{iq}\wedge(\omega_{\mathfrak k})^{p}_k\right]\wedge(\omega_{\mathfrak p})^i_l\ .
\end{align*}

\subsection{Field equations}

To find the field equations we have to find the contraction of $\Theta_\lambda$ with vertical vectors. We divide this task considering vertical vectors of the form $X+Y+Z$, where $X$ is $pr_0$-projectable and $\tau_1$-vertical, $Y$ is $p_1$-vertical, and $Z$ is $p_2$-vertical. Let us now give a useful description for the vectors $Y$ and $Z$. 

Going back to the notation of Section \ref{sec:tauto}, consider a cross section $\xi\in\Gamma(\bar\tau^k_{n,q})$. Then, since this is a vector bundle, we have the associated vertical lift, which is a $\bar\tau^k_{n,q}$-vertical vector field $\delta\xi\in\mathfrak X(\Lambda^k_{n,q}T^*P\otimes V)$ given by
\[
\delta\xi(\alpha_u)=\left.\frac{d}{dt}\right|_{t=0}(\alpha_u+t\xi(u)).
\]
In other words, this is the vector field whose flow is given by $\varphi^\xi_t(\alpha_u)=\alpha_u+t\xi(u)$, for every $\alpha_u\in E:=\Lambda^k_{n,q}T^*_uP\otimes V$. It is clear that $\delta\xi$ is $\bar\tau^k_{n,q}$-vertical, therefore it annihilates the tautological form $\Theta^k_{n,q}$ (because this is a horizontal form by definition). Now, let us see the contraction of this type of vectors with the differential $\dif\Theta^k_{n,q}$.
\begin{lem}
The following identity holds
\[
\delta\xi\lrcorner\dif\Theta^k_{n,q}=(\bar\tau^k_{n,q})^*(\xi)\ .
\]
\end{lem}
\begin{proof}
Let us compute
\begin{align*}
\delta\xi\lrcorner\dif\Theta^k_{n,q}&=\mathcal L_{\delta\xi}\Theta^k_{n,q}-\dif\delta\xi\lrcorner\Theta^k_{n,q}=\mathcal L_{\delta\xi}\Theta^k_{n,q}\\
&=\lim_{t\rightarrow 0}\frac{1}{t}\left[\left.\Theta^k_{n,q}\right|_{\varphi^\xi_t(\alpha_u)}\left(T_{\alpha_u}\varphi^\xi_t(\cdot),\dots,T_{\alpha_u}\varphi^\xi_t(\cdot)\right)-\left.\Theta^k_{n,q}\right|_{\alpha_u}(\cdot,\dots,\cdot)\right].
\end{align*}
Now, if $X_1,\dots,X_k\in T_{\alpha_u}E$,
\begin{align*}
&\left.\Theta^k_{n,q}\right|_{\varphi^\xi_t(\alpha_u)}\left(T_{\alpha_u}\varphi^\xi_t(X_1),\dots,T_{\alpha_u}\varphi^\xi_t(X_m)\right)\\
&=\varphi^\xi_t(\alpha_u)\left(T_{\varphi^\xi_t(\alpha_u)}\bar\tau^k_{n,q}\circ T_{\alpha_u}\varphi^\xi_t(X_1),\dots,T_{\varphi^\xi_t(\alpha_u)}\bar\tau^k_{n,q}\circ T_{\alpha_u}\varphi^\xi_t(X_m)\right)\\
&=(\alpha_u+t\xi(u))\left(T_{\alpha_u}(\bar\tau^k_{n,q}\circ \varphi^\xi_t)(X_1),\dots,T_{\alpha_u}(\bar\tau^k_{n,q}\circ \varphi^\xi_t)(X_m)\right)\\
&=(\alpha_u+t\xi(u))\left(T_{\alpha_u}\bar\tau^k_{n,q}(X_1),\dots,T_{\alpha_u}\bar\tau^k_{n,q}(X_m)\right).
\end{align*}
Then
\begin{align*}
\left.\mathcal L_{\delta\xi}\Theta^k_{n,q}\right|_{\alpha_u}(X_1,\dots,X_m)&=\lim_{t\rightarrow 0}\frac{1}{t}\left[(\alpha_u+t\xi(u))\left(T_{\alpha_u}\bar\tau^k_{n,q}(X_1),\dots,T_{\alpha_u}\bar\tau^k_{n,q}(X_m)\right)-\right.\\
&\left.\alpha_u\left(T_{\alpha_u}\bar\tau^k_{n,q}(X_1),\dots,T_{\alpha_u}\bar\tau^k_{n,q}(X_m)\right)\right]\\
&=\xi(u)\left(T_{\alpha_u}\bar\tau^k_{n,q}(X_1),\dots,T_{\alpha_u}\bar\tau^k_{n,q}(X_m)\right)\\
&=\left.(\bar\tau^k_{n,q})^*(\xi)\right|_{\alpha_u}(X_1,\dots,X_m)\qedhere
\end{align*}

\end{proof}

Let us start with the $p_2$-vertical vector fields. Let $Z=\delta\beta$ for some section $\beta\in\Gamma(p_2)$. Observe that the unique non-vanishing contraction for this kind of vectors is with the form $\dif\Theta_{ij}$, so
\begin{align*}
\delta\beta\lrcorner\Omega_\lambda&=\eta^{jl}(\delta\beta\lrcorner\dif\Theta_{ij}\wedge(\omega_{\mathfrak p})^i_l)
=\eta^{jl}pr_0^*(\beta_{ij})\wedge(\omega_{\mathfrak p})^i_l
=\eta^{jl}pr_0^*\left(\beta_{ij}\wedge(\omega_{\mathfrak p})^i_l\right).
\end{align*}
Now, if $\sigma\in\Gamma (\tau_\lambda)$ is a cross section, 
in order to be an extremal for our variational problem it must fulfill
\begin{equation}
(pr_0\circ\sigma)^*\left(\eta^{jl}\beta_{ij}\wedge(\omega_{\mathfrak p})^i_l\right)=0\ ,
\label{metric1}
\end{equation}
for every section $\beta$. Therefore, it must fulfill the equation
\[
(pr_0\circ\sigma)^*(\omega_{\mathfrak p})^i_l=0\ ,
\]
which is none other than the metricity condition.

\begin{note}
This can be seen locally as follows. Since Eq. \eqref{metric1} must hold for every $\beta_{ij}$, in particular it must hold for horizontal $\beta_{ij}$. Then, writing $\sigma^*\left[(\omega_{\mathfrak p})^i_l\eta^{jl}\right]=f^{ij}_\lambda\dif x^\lambda$ and $\beta_{ij}=p^\mu_{ij}\dif^{m-1}x_\mu$, for some functions $p^\mu_{ij}\in C^\infty(J^1\tau)$ and $f^{ij}_\lambda\in C^\infty(M)$ (both symmetric in the indices $ij$), we have
\[
\sigma^*\left(\beta_{ij}\wedge(\omega_{\mathfrak p})^i_l\eta^{jl}\right)=(-1)^{m-1}(p^\mu_{ij}\circ \sigma)f^{ij}_\lambda\dif x^\lambda\wedge\dif^{m-1}x_\mu=(-1)^{m-1}(p^\lambda_{ij}\circ \sigma)f^{ij}_\lambda=0 \ ,
\]
which implies
\[
(p^\lambda_{ij}\circ \sigma)f^{ij}_\lambda=0.
\]
then, varying $p^\mu_{ij}$ we conclude $f^{ij}_\lambda=0$ which means $\sigma^*(\omega_{\mathfrak p})^i_l\eta^{jl}=0$ and then $\sigma^*(\omega_{\mathfrak p})^i_l=0$.

\end{note}

Now consider $Y=\delta\gamma$ for some section $\gamma\in\Gamma( p_1)$. Similarly to the previous case, the unique non-vanishing contraction for this kind of vectors is that with the form $\dif\Theta_{l}$, so
\begin{align*}
\delta\gamma\lrcorner\Omega_\lambda&=T^l\wedge(\delta\gamma\lrcorner\dif\Theta_l)
=T^l\wedge pr_0^*(\gamma_l)
=pr_0^*(T^l\wedge \gamma_l).
\end{align*}
If $\sigma\in\Gamma (\tau_\lambda)$ is a cross section, then in order to be an extremal for the variational problem, it must fulfill
\[
(pr_0\circ\sigma)^*\left(T^l\wedge \gamma_l\right)=0\ ,
\]
for every section $\gamma$. Therefore, it must fulfill the equation
\[
(pr_0\circ\sigma)^*T^l=0\ ,
\]
which is in turn the zero torsion condition.

Moving on, we separate several cases for the vector fields $X$. Every vector field can be written as a sum of a $\tau_{10}$-vertical vector field plus a $\tau_{10}$-projectable $\tau$-vertical vector field. The first kind is generated by the vector fields $(\theta^a,(E^b_c)_{LM})^V$, while the second is generated by the set of infinitesimal generators of the action. Consider now a vector field $X=g^c_{ab}(\theta^a,(E^b_c)_{LM})^V$, for some functions $g^c_{ab}\in C^\infty(J^1\tau)$. We only consider those contractions that do not vanish when taking the pullback by a section fulfilling the previous conditions we have found so far. Then
\begin{align*}
g^c_{ab}(\theta^a,(E^b_c)_{LM})^V\lrcorner \Omega_\lambda&=g^c_{ab}\left(\theta^q\wedge\Theta_s+(-1)^{m-1}\eta^{tq}\Theta_{st}\right)\delta^s_c\delta^b_q\wedge\theta^a\\
&=g^c_{aq}\left(\theta^q\wedge\Theta_c+(-1)^{m-1}\eta^{tq}\Theta_{ct}\right)\wedge\theta^a \ .
\end{align*}
Therefore, taking the pullback by a solution to the previous equations and varying the functions $g^c_{ab}$, we conclude that it must also fulfill
\begin{align*}
\sigma^*\left(\Theta_c\wedge\theta^q-\eta^{tq}\Theta_{ct}\right)=0 \ ,
\end{align*}
or
\begin{align*}
\sigma^*\left(\eta_{qt}\Theta_c\wedge\theta^q-\Theta_{ct}\right)=0 \ ,
\end{align*}
and given the symmetry of $\Theta_{ct}$, we have
\[
\sigma^*\left(\eta_{qt}\Theta_c\wedge\theta^q\right)=\sigma^*\left(\Theta_{ct}\right)=0 \ ,
\]
which in turn gives us
\[
\sigma^*\left(\Theta_c\right)=0 \ .
\]
Now assume that $X=A_{J^1\tau}$. If $A\in\mathfrak k$, then it is easy to see that
\[
A_{J^1\tau}\lrcorner\Omega_\lambda=0 \ .
\]
so we have no new equations.

Finally, consider now $A\in\mathfrak p$. In this case
\begin{align*}
A_{J^1\tau}\lrcorner\Omega_\lambda&=\left[2r\eta^{qp}A^{t}_p\theta_{stI'J'}-\eta^{tq}A^{l}_{l}\theta_{stI'J'}\right]\wedge\Omega^{s}_q\wedge\Omega^{I'J'}\\
&=2rA^{t}_p\left[\eta^{qp}\theta_{stI'J'}-\frac{1}{2r}\eta^{kq}\delta^{p}_{t}\wedge\theta_{skI'J'}\right]\wedge\Omega^{s}_q\wedge\Omega^{I'J'} \ ,
\end{align*}
and simplifying this expression using the symmetry properties we arrive to the last equation
\[
\left[\theta_{stI'J'}\wedge\Omega^s_q+\theta_{sqI'J'}\wedge\Omega^s_t-\frac{1}{r}\eta_{tq}\left(\eta^{kl}\theta_{skI'J'}\wedge\Omega^s_l\right)\right]\wedge\Omega^{I'J'}=0 \ ,
\]
which can be written as
\[
\left(\Omega^{sa}\eta^{tb}+\Omega^{sb}\eta^{ta}-\frac{1}{r}\eta^{ab}\Omega^{st}\right)\wedge\theta_{stI'J'}\wedge\Omega^{I'J'}=0 \ .
\]
This is equivalent to the Euler-Lagrange equations associated with the Lovelock Lagrangian of order $r$ \cite{Lovelock1970} as the following lemma shows:
\begin{lem}
The vanishing of the pullback of the $m$-form
\[
\Psi^{ab}:=\left(\Omega^{sa}\eta^{tb}+\Omega^{sb}\eta^{ta}-\frac{1}{r}\eta^{ab}\Omega^{st}\right)\wedge\theta_{stI'J'}\wedge\Omega^{I'J'}
\]
by a local section $\sigma\in\Gamma \tau_1$ is equivalent to the Euler-Lagrange equations associated with the Lovelock Lagrangian.
\end{lem}
\begin{proof}
Take fibered coordinates and remember the relation between the coefficients of the curvature form and the curvature tensor [see \eqref{omegar}]. Then
\begin{align*}
e_a^\mu\Psi^{ab}e_b^\nu=&\left(e_a^\mu\Omega^{sa}\eta^{tb}e_b^\nu+e_a^\mu\Omega^{sb}\eta^{ta}e_b^\nu-\frac{1}{r}e_a^\mu\eta^{ab}e_b^\nu\Omega^{st}\right)\wedge\theta_{stI'J'}\wedge\Omega^{I'J'}\\
=&\left(e_a^\mu\Omega^{sa}e_\rho^t g^{\rho\nu}+\Omega^{sb}e_\rho^tg^{\rho\mu}e_b^\nu-\frac{1}{r}g^{\mu\nu}\Omega^{st}\right)\wedge\theta_{stI'J'}\wedge\Omega^{I'J'}\\
=&\left( e_\sigma^se_\tau^ae_\rho^t e_a^\mu R^{\sigma\tau}_{\theta\gamma}g^{\rho\nu}+e_\sigma^s e_\tau^ae_\rho^te_a^\nu R^{\sigma\tau}_{\theta\gamma}g^{\rho\mu}
-\frac{1}{r}g^{\mu\nu}e_\sigma^s e_\tau^tR^{\sigma\tau}_{\theta\gamma}\right)\wedge\dif x^\theta\wedge\dif x^\gamma\wedge\theta_{stI'J'}\wedge\Omega^{I'J'}\\
=&\ e_\sigma^sR^{\sigma\tau}_{\theta\gamma}\left(  e_\rho^t \delta_\tau^\mu g^{\rho\nu}+ e_\rho^t \delta_\tau^\nu g^{\rho\mu}-\frac{1}{r}g^{\mu\nu} e_\tau^t\right)\wedge\dif x^\theta\wedge\dif x^\gamma\wedge\theta_{stI'J'}\wedge\Omega^{I'J'}\ ,
\end{align*}
and renaming some indices
\begin{align*}
e_a^\mu\Psi^{ab}e_b^\nu=&e_{\lambda_1}^{i_1}R^{\lambda_1\theta_1}_{\alpha_1\beta_1}\left(  e_\rho^{j_1} \delta_{\theta_1}^\mu g^{\rho\nu}+ e_\rho^{j_1} \delta_{\theta_1}^\nu g^{\rho\mu}-\frac{1}{r}g^{\mu\nu} e_{\theta_1}^{j_1}\right)\wedge\dif x^{\alpha_1}\wedge\dif x^{\beta_1}\wedge\theta_{IJ}\wedge\Omega^{I'J'}\\
=&\frac{1}{(m-2r)!}\varepsilon_{i_1j_1\dots i_rj_rs_1\dots s_p}\varepsilon^{\alpha_1\beta_1\dots \alpha_r\beta_r\kappa_1\dots\kappa_p}e_{\lambda_1}^{i_1}\left(  e_\rho^{j_1} \delta_{\theta_1}^\mu g^{\rho\nu}+ e_\rho^{j_1} \delta_{\theta_1}^\nu g^{\rho\mu}-\frac{1}{r}g^{\mu\nu} e_{\theta_1}^{j_1}\right)\\
& e^{i_2}_{\lambda_2}e^{j_2}_{\theta_2}\dots e^{i_r}_{\lambda_r}e^{j_r}_{\theta_r}e^{s_1}_{\kappa_1}\dots e^{s_p}_{\kappa_p}R^{\lambda_1\theta_1}_{\alpha_1\beta_1}\dots R^{\lambda_r\theta_r}_{\alpha_r\beta_r}\dif^mx\\
=&\frac{\det(e)}{(m-2r)!}\varepsilon^{\alpha_1\beta_1\dots \alpha_r\beta_r\kappa_1\dots\kappa_p}\left[  \phantom{\frac{1}{r}}\varepsilon_{\lambda_1\rho\lambda_2\theta_2\dots\lambda_r\theta_r\kappa_1\dots \kappa_p} (\delta_{\theta_1}^\mu g^{\rho\nu}+ \delta_{\theta_1}^\nu g^{\rho\mu})\right.\\
&\left.-\frac{1}{r}g^{\mu\nu} \varepsilon_{\lambda_1\theta_1\lambda_2\theta_2\dots\lambda_r\theta_r\kappa_1\dots \kappa_p}\right]R^{\lambda_1\theta_1}_{\alpha_1\beta_1}\dots R^{\lambda_r\theta_r}_{\alpha_r\beta_r}\dif^mx\\
=&\det(e)\left[ \delta^{\alpha_1\beta_1\dots \alpha_r\beta_r}_{\lambda_1\rho\lambda_2\theta_2\dots\lambda_r\theta_r} (\delta_{\theta_1}^\mu g^{\rho\nu}+ \delta_{\theta_1}^\nu g^{\rho\mu})-\frac{1}{2r}\delta^{\alpha_1\beta_1\dots \alpha_r\beta_r}_{\lambda_1\theta_1\dots\lambda_r\theta_r}(g^{\mu\rho}\delta_\rho^\nu+g^{\nu\rho}\delta_\rho^\mu) \right]R^{\lambda_1\theta_1}_{\alpha_1\beta_1}\dots R^{\lambda_r\theta_r}_{\alpha_r\beta_r}\dif^mx.
\end{align*}

Now, as a consequence of the symmetry of $R^{\lambda_1\theta_1}_{\alpha_1\beta_1}\dots R^{\lambda_r\theta_r}_{\alpha_r\beta_r}$, we can write
\begin{align*}
e_a^\mu\Psi^{ab}e_b^\nu=&\frac{\det(e)}{r}\left[ \sum_{k=1}^r\delta^{\alpha_2\beta_2\dots \alpha_1\beta_1\dots \alpha_r\beta_r}_{\lambda_2\theta_2\dots\underbrace{\scriptstyle\lambda_1\rho}_{k\text{ position}}\dots\lambda_r\theta_r} (\delta_{\theta_1}^\mu g^{\rho\nu}+ \delta_{\theta_1}^\nu g^{\rho\mu})\right.\\
&\left.
\vphantom{\sum_{k=1}^r\delta^{\alpha_2\beta_2\dots \alpha_1\beta_1\dots \alpha_r\beta_r}_{\lambda_2\theta_2\dots\underbrace{\scriptstyle\lambda_1\rho}_{k\text{ position}}\dots\lambda_r\theta_r}}
-\frac{1}{2}\delta^{\alpha_1\beta_1\dots \alpha_r\beta_r}_{\lambda_1\theta_1\dots\lambda_r\theta_r}(g^{\mu\rho}\delta_\rho^\nu+g^{\nu\rho}\delta_\rho^\mu) \right]R^{\lambda_1\theta_1}_{\alpha_1\beta_1}\dots R^{\lambda_r\theta_r}_{\alpha_r\beta_r}\dif^mx \ .
\end{align*}
Observe that, renaming indices and using the symmetries of the curvature tensor,
\begin{align*}
&\sum_{k=1}^r\delta_{\theta_1}^\mu\delta^{\alpha_2\beta_2\dots \alpha_1\beta_1\dots \alpha_r\beta_r}_{\lambda_2\theta_2\dots\underbrace{\scriptstyle\lambda_1\rho}_{k\text{ position}}\dots\lambda_r\theta_r}  g^{\rho\nu}R^{\lambda_1\theta_1}_{\alpha_1\beta_1}\dots R^{\lambda_r\theta_r}_{\alpha_r\beta_r}=
\sum_{k=1}^r\delta_{\theta_k}^\mu\delta^{\alpha_1\beta_1\dots \alpha_k\beta_k\dots \alpha_r\beta_r}_{\lambda_1\theta_1\dots\underbrace{\scriptstyle\lambda_k\rho}_{k\text{ position}}\dots\lambda_r\theta_r}  g^{\rho\nu}R^{\lambda_1\theta_1}_{\alpha_1\beta_1}\dots R^{\lambda_r\theta_r}_{\alpha_r\beta_r}\\
=&\frac{1}{2}\sum_{k=1}^r\left(\delta_{\theta_k}^\mu\delta^{\alpha_1\beta_1\dots \alpha_k\beta_k\dots \alpha_r\beta_r}_{\lambda_1\theta_1\dots\underbrace{\scriptstyle\lambda_k\rho}_{k\text{ position}}\dots\lambda_r\theta_r} +\delta_{\lambda_k}^\mu\delta^{\alpha_1\beta_1\dots \alpha_k\beta_k\dots \alpha_r\beta_r}_{\lambda_1\theta_1\dots\underbrace{\scriptstyle\rho\theta_k}_{k\text{ position}}\dots\lambda_r\theta_r} \right) g^{\rho\nu}R^{\lambda_1\theta_1}_{\alpha_1\beta_1}\dots R^{\lambda_r\theta_r}_{\alpha_r\beta_r} \ .
\end{align*}
Hence, we get
\begin{align*}
e_a^\mu\Psi^{ab}e_b^\nu=&\frac{\det(e)}{2r}\left[ \sum_{k=1}^r\left(\delta_{\theta_k}^\mu\delta^{\alpha_1\beta_1\dots \alpha_k\beta_k\dots \alpha_r\beta_r}_{\lambda_1\theta_1\dots\underbrace{\scriptstyle\lambda_k\rho}_{k\text{ position}}\dots\lambda_r\theta_r} +\delta_{\lambda_k}^\mu\delta^{\alpha_1\beta_1\dots \alpha_k\beta_k\dots \alpha_r\beta_r}_{\lambda_1\theta_1\dots\underbrace{\scriptstyle\rho\theta_k}_{k\text{ position}}\dots\lambda_r\theta_r} \right) g^{\rho\nu}\right.\\
&\left.
\vphantom{\sum_{k=1}^r\delta^{\alpha_2\beta_2\dots \alpha_1\beta_1\dots \alpha_r\beta_r}_{\lambda_2\theta_2\dots\underbrace{\scriptstyle\lambda_1\rho}_{k\text{ position}}\dots\lambda_r\theta_r}}
-\delta_\rho^\mu\delta^{\alpha_1\beta_1\dots \alpha_r\beta_r}_{\lambda_1\theta_1\dots\lambda_r\theta_r}g^{\nu\rho}+\left(\mu\leftrightarrow \nu\right) \right]R^{\lambda_1\theta_1}_{\alpha_1\beta_1}\dots R^{\lambda_r\theta_r}_{\alpha_r\beta_r}\dif^mx\\
=&-\frac{\det(e)}{2r}\left( \delta^{\mu\alpha_1\beta_1\dots \alpha_r\beta_r}_{\rho\lambda_1\theta_1\dots\lambda_r\theta_r}g^{\rho\nu} +\delta^{\nu\alpha_1\beta_1\dots \alpha_r\beta_r}_{\rho\lambda_1\theta_1\dots\lambda_r\theta_r}g^{\rho\mu}\right)R^{\lambda_1\theta_1}_{\alpha_1\beta_1}\dots R^{\lambda_r\theta_r}_{\alpha_r\beta_r}\dif^mx.
\end{align*}
Therefore, the vanishing of $\Psi^{ab}$ is equivalent to the equations
\[
\left( \delta^{\mu\alpha_1\beta_1\dots \alpha_r\beta_r}_{\rho\lambda_1\theta_1\dots\lambda_r\theta_r}g^{\rho\nu} +\delta^{\nu\alpha_1\beta_1\dots \alpha_r\beta_r}_{\rho\lambda_1\theta_1\dots\lambda_r\theta_r}g^{\rho\mu}\right)R^{\lambda_1\theta_1}_{\alpha_1\beta_1}\dots R^{\lambda_r\theta_r}_{\alpha_r\beta_r}=0\ ,
\]
which are indeed the Euler-Lagrange equations described by Lovelock in \cite{Lovelock1970}.
\end{proof}

\begin{note}
The tensor
\[
A^{\mu\nu}=\left( \delta^{\mu\alpha_1\beta_1\dots \alpha_r\beta_r}_{\rho\lambda_1\theta_1\dots\lambda_r\theta_r}g^{\rho\nu} +\delta^{\nu\alpha_1\beta_1\dots \alpha_r\beta_r}_{\rho\lambda_1\theta_1\dots\lambda_r\theta_r}g^{\rho\mu}\right)R^{\lambda_1\theta_1}_{\alpha_1\beta_1}\dots R^{\lambda_r\theta_r}_{\alpha_r\beta_r}
\]
is the unique tensor (up to a constant) fulfilling the following properties:
\begin{itemize}
\item It depends only on the metric tensor $g_{\alpha\beta}$ and its first two derivatives,
\item it is symmetric in $\mu\nu$, and
\item it has vanishing covariant derivative with respect to the Levi-Civita connection.
\end{itemize}
This was introduced by Lovelock \cite{Lovelock1970,lovel} as a generalization of the Einstein tensor for dimensions higher than 4. In this sense, the form $\Psi$ is a global expression for such a tensor.
\end{note}

Gathering all the equations found so far, we conclude that the solutions to the field equations are those cross sections that are integral for the following exterior differential system
\begin{equation}\label{eq:EDSForLovelock}
  \mathcal J:=\left<\Theta_l,\Theta_{ij},T^l,(\omega_{\mathfrak p})^i_j,(\Omega_{\mathfrak p})^i_j,\Omega^k_l\wedge\theta^l,\Psi^{ij}\right>.
\end{equation}

Notice that the forms $\Omega^k_l\wedge\theta^l$ and $(\Omega_{\mathfrak p})^i_j$ come as a consequence of the first Bianchi identity and the structure equation for the connection $\omega$, respectively.

\begin{note}
It is interesting to note that the forms $\Psi^{ij}$ can be written in terms of the $(m-1)$-forms $\theta_{lIJ}\wedge\Omega^{IJ}$. Indeed
\begin{align*}
(\eta^{il}\theta^j)\wedge \theta_{lIJ}\wedge\Omega^{IJ}=&(\eta^{il}\theta^j)\wedge\Omega^{st}\wedge\theta_{lstI'J'}\wedge\Omega^{I'J'}=\eta^{il}\wedge\Omega^{st}\wedge(\theta^j\theta_{lstI'J'})\wedge\Omega^{I'J'}\\
=&-\Omega^{st}\wedge\eta^{il}\left[
\vphantom{\sum_{a=1}^{r-1}}
-\delta^j_l\theta_{stI'J'}+\delta^j_s\theta_{ltI'J'}-\delta^j_t\theta_{slI'J'}
\right.\\ &\left.
-\sum_{a=1}^{r-1}\left(\delta^j_{i_a}\theta_{lstI'_aJ'}-\delta^j_{j_a}\theta_{lstI'J'_a}\right)\right]\wedge\Omega^{I'J'},
\end{align*}
where the symbol $I'_a$ has been used to indicate that the index $i_a$ has been removed from the multiindex $I'$. Then rearranging terms and using the symmetry of the product $\Omega^{st}\wedge\Omega^{I'J'}$,
\begin{align*}
(\eta^{il}\theta^j)\wedge \theta_{lIJ}\wedge\Omega^{IJ}=&\eta^{ij}\theta_{IJ}\wedge\Omega^{IJ}-\eta^{il}\left[\delta^j_s\theta_{ltI'J'}-\delta^j_t\theta_{lsI'J'}
\vphantom{\sum_{a=1}^{r-1}}
\right.\\&\left.
-\sum_{a=1}^{r-1}\left(\delta^j_{i_a}\theta_{lstI'_aJ'}-\delta^j_{j_a}\theta_{lstI'J'_a}\right)\right]\wedge\Omega^{st}\wedge\Omega^{I'J'}\\
=&\eta^{ij}\theta_{IJ}\wedge\Omega^{IJ}-r\eta^{il}\left(\delta^j_s\theta_{ltI'J'}-\delta^j_t\theta_{lsI'J'}\right)\wedge\Omega^{st}\wedge\Omega^{I'J'}\\
=&\eta^{ij}\theta_{IJ}\wedge\Omega^{IJ}-r\eta^{il}\left(\theta_{ltI'J'}\wedge\Omega^{jt}-\theta_{lsI'J'}\wedge\Omega^{sj}\right)\wedge\Omega^{I'J'},
\end{align*}
and renaming some indices we get
\[
(\eta^{il}\theta^j)\wedge \theta_{lIJ}\wedge\Omega^{IJ}=\eta^{ij}\theta_{IJ}\wedge\Omega^{IJ}-2r\eta^{it}\theta_{stI'J'}\wedge\Omega^{sj}\wedge\Omega^{I'J'}.
\]
Therefore
\[
-\frac{1}{2r}\left(\eta^{il}\theta^j+\eta^{jl}\theta^i\right)\wedge \theta_{lIJ}\wedge\Omega^{IJ}=\left(\eta^{it}\Omega^{sj}+\eta^{jt}\Omega^{si}-\frac{1}{r}\eta^{ij}\Omega^{st}\right)\wedge\theta_{stI'J'}\wedge\Omega^{I'J'}
\]
from where it follows
\[
\Psi^{ij}=-\frac{1}{2r}\left(\eta^{il}\theta^j+\eta^{jl}\theta^i\right)\wedge \theta_{lIJ}\wedge\Omega^{IJ}.
\]
\end{note}

\begin{note}
In particular, in \cite{Capriotti:2012gg,doi:10.1142/S0219887818500445}, two exterior differential systems where introduced to describe Palatini gravity, the first as a Griffiths variational problem takes the forms $\theta_{lIJ}\wedge\Omega^{IJ}$ among its generators, and the second one by means of the unified formalism associated with the first one, where the forms $\Psi^{ij}$ are used instead.

\end{note}


\section{Conclusions and outlook}
\label{concl}

We have defined the Lovelock Lagrangian in the context of the multisymplectic framework for classical field theories,
and we have used this geometric formulation 
to characterize the properties of this Lagrangian,
to establish its Griffiths variational problem and derive the
corresponding field equations,
and to study the infinitesimal symmetries
of the system.
As this Lagrangian is singular, this is a (pre-multisymplectic) field theory with constraints and then we have developed the
Lagrangian--Hamiltonian unified formalism 
which is very suitable for its analysis.

Furthermore, if a variational problem has constraints, one can consider applying the constraints before or after performing the variations. In general, these two procedures lead to different sets of equations \cite{BK-86}. When both sets of equations are equal, we say that the variational problem has a \emph{consistent truncation} (by the constraints). In \cite{DP2} the authors claim that the Lovelock Lagrangians can be characterized by the consistency of the Levi-Civita truncation; that is, replacing the arbitrary connection by the Levi-Civita connection associated to the metric. We hope that the formalism presented in this paper could be very appropriate to analyze this property of gravitational theories, and to study the concept of consistent truncation for variational principles in a geometrical way. 
This would be a topic for further research.

Finally, the methods and results obtained in this paper are suitable to describe other gravity theories, such as
the $f(R)$ and $f(T)$ \cite{CdL-2011,Fe-12} models or the BF-gravity \cite{CGM-16}. 
Thus, the multisymplectic formulation of these theories and, eventually,
the development of their Lagrangian-Hamiltonian unified formalism 
are lines of further research.


\appendix
\label{apendice}

\section{Geometric elements}

\subsection{Levi-Civita symbol and generalized Kronecker delta}\label{sec:levi}

We denote the Levi-Civita symbol in $k$ indices by $\varepsilon^{i_1\dots i_k}$ and $\varepsilon_{i_1\dots i_k}$
\[
\varepsilon^{i_1\dots i_k}=
\begin{cases}
1 & \text{if }(i_1,\dots, i_k)\text{ is an even permutation of }(1,\dots, k),\\
-1 & \text{if }(i_1,\dots, i_k)\text{ is an odd permutation of }(1,\dots, k),\\
0 & \text{ in other case}
\end{cases}
\]
On the other hand, the generalized Kronecker delta \cite{book:7577} in $k$ indices $\delta^{i_1\dots i_k}_{j_1\dots j_k}$ is given by
\[
\delta^{i_1\dots i_k}_{j_1\dots j_k}=
\begin{cases}
1 & \text{if }(i_1,\dots, i_k)\text{ is an even permutation of }(j_1,\dots, j_k),\\
-1 & \text{if }(i_1,\dots, i_k)\text{ is an odd permutation of }(j_1,\dots, j_k),\\
0 & \text{ in other case}
\end{cases}
\]
and can also be expressed as
\[
\delta^{i_1\dots i_k}_{j_1\dots j_k}=
\begin{vmatrix*}
\delta^{i_1}_{j_1} & \dots & \delta^{i_k}_{j_1}\\
\vdots & \ddots & \vdots \\
\delta^{i_1}_{j_k} & \dots & \delta^{i_k}_{j_k}
\end{vmatrix*} \ .
\]

Both tensors are completely antisymmetric in their indices and they are seen to fulfill the following properties
\begin{enumerate}
\item $\varepsilon_{i_1\dots i_ki_{k+1}\dots i_m}\varepsilon^{i_1\dots i_kj_{k+1}\dots j_m}=k!\delta_{i_{k+1}\dots i_m}^{j_{k+1}\dots j_m}$.
\item For any tensor $a^{\nu_1\dots \nu_p}$,
\[
\frac{1}{p!}\delta^{\mu_1\dots \mu_p}_{\nu_1\dots \nu_p}a^{\nu_1\dots \nu_p}=a^{[\nu_1\dots \nu_p]}.
\]
\item If $A$ is a matrix with entries $a^i_j$,
\[
\varepsilon_{i_1\dots i_m}a^{i_1}_{j_1}\dots a^{i_m}_{j_m}=\det(A)\varepsilon_{j_1\dots j_m}.
\]
\end{enumerate}
Similar properties hold interchanging lower and upper indices.

\subsection{Hodge star operator}\label{sec:hodge}

Let $V$ be an $m$-dimensional real vector space and $\eta$ a non-degenerate bilinear symmetric form on $V$. For each $k\leq m$, we can define another non-degenerate bilinear and symmetric form $\hat\eta$ on $\Lambda^kV$ extending $\eta$ as the unique bilinear form such that on elementary $k$-vectors $\alpha=\alpha_1\wedge\dots\wedge\alpha_k$ and $\beta=\beta_1\wedge\dots\wedge\beta_k$,
\[
\hat\eta(\alpha,\beta):=\det[\eta(\alpha_i,\beta_j)].
\]

Given that $\Lambda^mV$ is one-dimensional, it follows that there are exactly two $m$-vectors $v$ fulfilling $\hat\eta(v,v)=1$. Let $\omega$ be a preferred unit $m$-vector (observe that fixing such a vector amounts to choosing an orientation for $V$). Then, the star Hodge operator $\star:\Lambda^kV\rightarrow\Lambda^{n-k}V$ related to $\eta$ is defined by requiring
\[
\alpha\wedge(\star\beta)=\hat\eta(\alpha,\beta)\omega.
\]
It is customary to refer to $\star\beta$ as the Hodge dual of $\beta$. 

Given an ordered basis $\{e_1,\dots,e_m\}$ of $V$ such that $\omega=e_1\wedge\dots\wedge e_m$, it is clear that $$\star(e_{i_1}\wedge\dots\wedge e_{i_k})=e_{i_{k+1}}\wedge\dots\wedge e_{i_m}$$ if and only if, $(i_1,\dots,i_m)$ is an even permutation of $(1,\dots,m)$.

\begin{lem}
If $\beta=\beta^{i_1\dots i_k}e_{i_1}\wedge\dots\wedge e_{i_k}$ ($k\leq m$). Then
\[
\star\beta=\frac{1}{(m-k)!}\beta^{i_1\dots i_k}e_{i_1}\eta_{i_1j_1}\dots\eta_{i_kj_k}\varepsilon^{j_{1}\dots j_m}e_{j_{k+1}}\wedge\dots\wedge e_{j_m}.
\]
\end{lem}
\begin{proof}
Let $\alpha=\alpha^{l_1\dots l_k}e_{l_{1}}\wedge\dots\wedge e_{l_k}$. Then
\begin{align*}
\alpha\wedge(\star\beta)&=\alpha^{l_1\dots l_k}e_{l_{1}}\wedge\dots\wedge e_{l_k}\wedge\left[\frac{1}{(m-k)!}\beta^{i_1\dots i_k}e_{i_1}\eta_{i_1j_1}\dots\eta_{i_kj_k}\varepsilon^{j_1\dots j_m}
e_{j_{k+1}}\wedge\dots\wedge e_{j_m}\vphantom{\frac{1}{(m-k)!}}\right]\\
&=\frac{1}{(m-k)!}\alpha^{l_1\dots l_k}e_{l_{1}}\beta^{i_1\dots i_k}\eta_{i_1j_1}\dots\eta_{i_kj_k}\varepsilon^{j_1\dots j_m}
e_{l_{1}}\wedge\dots\wedge e_{l_k}\wedge e_{j_{k+1}}\wedge\dots\wedge e_{j_m}\\
&=\frac{1}{(m-k)!}\alpha^{l_1\dots l_k}e_{l_{1}}\beta^{i_1\dots i_k}\eta_{i_1j_1}\dots\eta_{i_kj_k}\varepsilon^{j_1\dots j_kj_{k+1}\dots j_m}\varepsilon_{l_1\dots l_kj_{k+1}\dots j_m}
e_1\wedge\dots\wedge e_m\\
&=\alpha^{l_1\dots l_k}e_{l_{1}}\beta^{i_1\dots i_k}\eta_{i_1j_1}\dots\eta_{i_kj_k}\delta^{j_1\dots j_k}_{l_1\dots l_k}e_1\wedge\dots\wedge e_m\\
&=\hat\eta(\alpha,\beta)e_1\wedge\dots\wedge e_m \ .
\end{align*}
\end{proof}

\subsection{Cartan decomposition of $\mathfrak{gl}(m)$}\label{sec:cartan}

Let us use the matrix $\eta$ in order to decompose $\mathfrak{gl}\left(m\right)$; in order to get it, consider the involution
\[
  \theta:\mathfrak{gl}\left(m,\mC\right)\rightarrow\mathfrak{gl}\left(m,\mC\right):A\mapsto-\eta A^\dagger\eta;
\]
the eigenspaces of $\theta$, associated to the eigenvalues $\pm1$, induce the decomposition 
\[
\mathfrak{gl}\left(m,\mC\right)=\mathfrak{u}\left(m-s,s\right)\oplus\mathfrak{s}\left(m-s,s\right)
\]
where $s$ is the signature of $\eta$. It should be noted that
\[
\left[\mathfrak{u}\left(m-s,s\right),\mathfrak{u}\left(m-s,s\right)\right]\subset\mathfrak{u}\left(m-s,s\right),\qquad\left[\mathfrak{s}\left(m-s,s\right),\mathfrak{s}\left(m-s,s\right)\right]\subset\mathfrak{u}\left(m-s,s\right),
\]
and that $\mathfrak{s}\left(m-s,s\right)$ is an invariant subspace under the adjoint action of $\mathfrak{u}\left(m-s,s\right)$.

This decomposition descends to $\mathfrak{gl}\left(m\right)\subset\mathfrak{gl}\left(m,\mC\right)$, namely
\[
\mathfrak{gl}\left(m\right)=\kf\oplus\pf,
\]
where
\[
\kf:=\mathfrak{u}\left(m-s,s\right)\cap\mathfrak{gl}\left(m\right),\qquad\pf:=\mathfrak{s}\left(m-s,s\right)\cap\mathfrak{gl}\left(m\right).
\]
Denoting $s:=\theta_{|\mathfrak{gl}\left(m\right)}$, we have that $\kf$ (resp. $\pf$) is the eigenspace corresponding to the eigenvalue $+1$ (resp. $-1$) for $s$. The projectors in every of these eigenspaces become
\[
\pi_\kf\left(A\right):=\frac{1}{2}\left(A-\eta A^T\eta\right),\qquad\pi_\pf\left(A\right):=\frac{1}{2}\left(A+\eta A^T\eta\right).
\]

Given  $N$ a manifold and $\gamma\in\Omega^p\left(N,\mathfrak{gl}\left(m\right)\right)$, we define
\[
\gamma_\kf:=\pi_\kf\circ\gamma,\qquad\gamma_\pf:=\pi_\pf\circ\gamma.
\]

If $\gamma=\gamma^i_jE^j_i$ is the expression of $\gamma$ in terms of the canonical basis of $\mathfrak{gl}\left(m\right)$, then we have
\[
\left(\gamma_\kf\right)^i_j=\frac{1}{2}\left(\gamma^i_j-\eta_{jp}\gamma^p_q\eta^{qi}\right)\qquad\text{and}\qquad \left(\gamma_\pf\right)^i_j=\frac{1}{2}\left(\gamma^i_j+\eta_{jp}\gamma^p_q\eta^{qi}\right).
\]
These previous considerations are useful when dealing with $\mathfrak{gl}\left(m\right)$-valued forms.

\subsection{Vector-valued and Lie-algebra-valued differential forms}\label{sec:vvalform}

Let $U,V$ and $W$ be finite dimensional real vector spaces and $M$ a smooth manifold and let $B:U\times V\rightarrow W$ be a bilinear map. If $\alpha\in\Omega^k(M,U)$ and $\beta \in\Omega^l(M,V)$ are differential forms with values in $U$ and $V$, respectively; we can define a new differential form with values in $W$, $B\left(\alpha\wedcol\beta\right)\in \Omega^{k+l}(M,W)$, as
\[
B\left(\alpha\wedcol\beta\right)(X_1,\dots,X_{k+l})=\frac{1}{k!l!}\sum_{\sigma\in S_{k+l}}\text{sgn}(\sigma)B\left(\alpha(X_{\sigma(1)},\dots,X_{\sigma(k)}),\beta(X_{\sigma(k+1)},\dots,X_{\sigma(k+l)})\right).
\]

We are interested in a series of particular instances of this definition:
\begin{description}
\item[Pairing.] Consider $U=V^*$, $W=\mR$ and $B=\left<\cdot,\cdot\right>$ the natural pairing between $V^*$ and $V$. In this particular case we denote 
\[
B\left(\alpha\wedcol\beta\right)=\left<\alpha\wedcol\beta\right>.
\]

\item[Linear representation.] Consider $U=\text{End}(V)$, $W=V$, and denote $B$ the natural action of $\text{End}(V)$ on $V$. In this particular case we denote 
\[
B\left(\alpha\wedcol\beta\right)=\alpha\weddot\beta.
\]

\item[Lie bracket.] Consider $U=V=W=\mathfrak g$, where $\mathfrak g$ is a Lie algebra and $B=[\cdot,\cdot]$ is the related Lie bracket. In this particular case we denote 
\[
B\left(\alpha\wedcol\beta\right)=[\alpha\wedcol\beta].
\]

\item[Wedge product.] Consider $U=V$, $W=\Lambda^2V$ and let $B=\cdot\wedge\cdot$ be the usual wedge product. In this particular case we denote 
\[
B\left(\alpha\wedcol\beta\right)=\alpha\wedge\beta.
\]

\item[Constant linear map.] Consider a zero form  $\alpha$ assigning to each $x\in M$ the same linear map $A\in \text{Lin}(V,W)$, in this particular case we denote
\[
B\left(\alpha\wedcol\beta\right)=A(\beta).
\]
\end{description}

Similar definitions can be given if $V$ is a vector bundle over $M$ and $\alpha$ and $\beta$ are vector bundle valued differential forms.

\section*{Acknowledgments}

J. Gaset and N. Rom\'an-Roy acknowledge
the financial support from the 
{\sl Spanish Ministerio de Econom\'{\i}a y Competitividad}
project MTM2014--54855--P, 
the {\sl  Spanish Ministerio de Ciencia, Innovaci\'on y Universidades} project
PGC2018-098265-B-C33, and
the {\sl Secretary of University and Research of the Ministry of Business and Knowledge of
the Catalan Government} project
2017--SGR--932.

\noindent S. Capriotti and L.M. Salomone thank CONICET for its financial support.

\noindent L.M. Salomone thanks Prof. N. Rom\'{a}n-Roy and the staff of the Department of Mathematics at UPC for their generous hospitality during his stay in Barcelona.

\noindent Finally, we thank the referee for his constructive comments and suggestions.






\addcontentsline{toc}{section}{References}
\itemsep 0pt plus 1pt

\end{document}